\documentclass[11pt]{article}

\usepackage{fullpage}
\usepackage{url}
\usepackage{smile}
\usepackage{pdfpages}
\usepackage{kpfonts}
\usepackage{dsfont}
\usepackage{tgpagella}
\usepackage{natbib}
\usepackage{algorithm,caption}
\usepackage[]{hyperref}
\usepackage[left]{lineno}
\usepackage{blindtext}
\usepackage[ruled, vlined, algo2e]{algorithm2e}
\usepackage{physics}
\usepackage{float}

\newcommand\figcaption{\def\@captype{figure}\caption}
\newcommand\tabcaption{\def\@captype{table}\caption}

\newcommand{\widesim}[2][1.5]{
  \mathrel{\overset{#2}{\scalebox{#1}[1]{$\sim$}}}
}
\DeclareMathAlphabet{\mathsf}{OT1}{cmss}{m}{n}
\SetMathAlphabet{\mathsf}{bold}{OT1}{cmss}{bx}{n}

\makeatletter

\newcommand{\Rmnum}[1]{\expandafter\@slowromancap\romannumeral #1@}
\makeatother

\begin{document}

\title{\huge Online Quantification of Input Model Uncertainty by Two-Layer Importance Sampling }

\date{}

\author{Tianyi Liu and Enlu Zhou\thanks{ T. Liu and E. Zhou are affiliated with School of Industrial and Systems Engineering at Georgia Tech. Liu's email address is tianyiliu@gatech.edu. Zhou's email address is enlu.zhou@isye.gatech.edu.}}
\maketitle
\begin{abstract}
Stochastic simulation has been widely used to analyze the performance of complex stochastic systems and facilitate decision making in those systems. Stochastic simulation is driven by the input model, which is a collection of probability distributions that model the stochasticity in the system. The input model is usually estimated using a finite amount of data, which introduces the so-called input model uncertainty to the simulation output. How to quantify input uncertainty has been studied extensively, and many methods have been proposed for the batch data setting, i.e., when all the data are available at once. However, methods for  ``streaming data'' arriving sequentially in time are still in demand, despite that streaming data have become increasingly prevalent in modern applications. To fill this gap, we propose a two-layer importance sampling framework that incorporates streaming  data for online input uncertainty quantification. Under this framework, we develop two algorithms that suit different application scenarios: the first scenario is when data come at a fast speed and there is no time for any new simulation in between updates; the second is when data come at a moderate speed and a few but limited simulations are allowed at each time stage.  We prove the consistency and asymptotic convergence rate results, which theoretically show the efficiency of our proposed approach. We further demonstrate the proposed algorithms on a numerical example of the news vendor problem.
\end{abstract}

\section{Introduction}
For a complex stochastic system, real-world experiments are usually expensive or difficult to conduct. In this case, stochastic simulation is always a powerful tool to analyze the system behavior. Stochastic simulation is driven by the input model, which is a collection of distributions that model the randomness in the system. There are generally two sources of uncertainty in a simulation experiment. One is the simulation uncertainty that reflects the intrinsic randomness of the system. The other is the input model uncertainty (or simply as input uncertainty), which is caused by the estimation of the input model from finite realizations of the  real-world stochastic processes. For example, when simulating a queuing network, we generate samples of customer arrivals and service times from appropriate distributions (input models). The simulation output (e.g., average queue length) depends on the parameters of the input model (e.g., arrival and service rates). These parameters are usually estimated from a finite amount of data and have estimation error (and hence input uncertainty). Without quantifying the input uncertainty, simulation users can hardly separate the input uncertainty from the simulation uncertainty, which can result in a wrong interpretation of simulation results. A proper quantification of input uncertainty can also provide inferences on system sensitivity or robustness to input uncertainty.

Various input uncertainty quantification methods have been proposed under batch data setting where data are available all at once. These methods include  Bayesian methods (e.g.,\cite{chick2001input,zouaoui2003accounting,zouaoui2004accounting,xie2014bayesian}), frequentist methods (e.g., \cite{barton1993uniform,barton2001resampling,cheng1997sensitivity}), delta methods (e.g., \cite{cheng1997sensitivity}),  meta-model assisted methods  (e.g., \cite{barton2013quantifying,xie2014bayesian}), and {some more recent ones (e.g., \cite{lam:2016sensitivity, ZhouZhu:2015risk, lam2016empirical, lam2018subsampling, lin2015single,feng2019WSCuq}). For a comprehensive review on input uncertainty quantification, the reader can refer to \cite{barton2012tutorial} and \cite{song2017input}.}

Despite the abundance of methods developed under the batch data setting, there is no method specifically designed to work with streaming data, which refer to data arriving one by one or in mini batches sequentially in time. With streaming data, it is natural to take an ``online'' approach to input uncertainty quantification, i.e., 
to update the quantification after each new input data point comes in. Online quantification can further facilitate online decision making in data-driven applications. For example,  in supply chain management, a retailer can continuously collect customer demand data to update the estimate of the demand distribution and quantify its uncertainty, which in turn leads to updated evaluation and improvement of the restocking policy.  

The main bottleneck to online quantification is the long running time of simulation experiments. Consider a naive approach that extends a batch method to the streaming data setting: we just repeat the method every time when new data point becomes available. This naive extension is obviously  inefficient since we have to re-run simulation experiments every time, and can hardly be used in practice especially when fast decision making is required. To have an applicable online method, we should be able to update estimates fast enough every time when new data point arrives, which means we can only do very few or even no simulation experiments every time. On the other hand, we need to maintain the accuracy of estimates over a long time horizon, which is much harder than controlling estimation accuracy in one time stage (e.g., under the batch data setting) due to the possibility that estimation error gets accumulated over time.  These two contradicting requirements (i.e., few simulation experiments every time and maintaining estimation accuracy over time) make the online problem extremely challenging.

To address the challenge described above, we need to make full use of every simulation experiment that has ever been run. This is the motivation behind our proposed approach of Two-Layer Importance Sampling (TLIS). More specifically, we assume the input distribution takes a parametric form with a known distribution form but unknown input parameter. Given a pre-specified prior distribution of the input parameter, we apply the Bayesian rule to update the posterior distribution when new data arrive. In this model, there are two layers of uncertainty. The outer-layer input uncertainty is characterized by the posterior distribution of the input parameter, and the inner-layer simulation uncertainty is caused by sampling from an input distribution. Our proposed approach applies the importance sampling  technique to reuse simulation outputs at both layers. At the outer layer, we reuse the system performance estimates under the input parameter samples from  posterior distributions at previous time stages. At the inner layer, we apply importance sampling to estimate the system performance under any new input parameter sample by a weighted average of all simulation outputs under different input parameters. This two-layer importance sampling application makes it possible to run very few or even no new simulation experiments at each time stage.

Our TLIS approach can be applied to scenarios with different speeds of data arrivals and quantification needs. In the first scenario of fast data arrival, we assume no simulation experiment is allowed between data arrivals. One example of this scenario is the stock market where the stock price changes rapidly and there is a need to keep updating quantification of risk. During market opening, there is no time for new simulation experiments, but a large amount of simulation experiments can be done during market closure, which can be viewed as initialization of the process. For this scenario, we apply TLIS by using the inner-layer importance sampling to estimate system performance with those simulation experiments done at the initial time stage. In the second scenario of moderate data arrival, we assume a small amount of simulation experiments are allowed at each time stage. An example is daily inventory management, where customer demand data is collected daily and used to update  evaluation of the restocking policy. Although the algorithm for the first scenario is still applicable here, we will take advantage of new simulation experiments and apply the inner-layer importance sampling to use new simulation outputs under input parameter samples drawn from the current posterior. 

We note our proposed TLIS is related to ``green simulation'' proposed by \cite{feng:2015WSCgreensim} and \cite{Feng2017GreenSR} in the sense of reusing simulation outputs from previous experiments. However, their key condition for convergence that the stationary measure exists does not hold in this online setting, where the posterior distribution changes with data arrivals. We also note a recent work \cite{feng2019WSCuq} applies green simulation for input uncertainty quantification, and they independently developed a method that is similar to our inner-layer importance sampling technique. However, they consider the batch data setting, while we consider streaming data as well as different scenarios of data arrival speeds. 

Our convergence results show that TLIS asymptotically tracks the true quantification over time. {Moreover, TLIS has an increased asymptotic convergence rate due to the importance sampling used in both outer-layer and inner-layer. Compared with the Direct Monte Carlo method (described in Section \ref{Direct_MC}), the outer-layer importance sampling improves the outer-layer convergence rate by a factor $O(1/\sqrt{K}),$ and the inner-layer importance sampling improves the inner-layer convergence rate by a factor $O(1/\sqrt{M}),$ where $K$ is the number of time stages we reuse the simulation outputs and  $M$ is the number of input parameter samples we draw at each time stage. } Numerical testing verifies our theoretical results and further shows the advantage of TLIS compared to the Direct Monte Carlo method, a Simple Importance Sampling method, and the online application of the Green Simulation method. 

We conclude our contribution as follows: 1) we are among the first to consider input uncertainty quantification with streaming data, and developed a Two-Layer Importance Sampling approach that can adapt to different speeds of data arrivals;
2) we theoretically and numerically showed that our algorithms over-perform  the online extension of some batch methods under the same simulation budget; 3) we proved some important properties of exponential families of distributions, which is of independent interest.

\section{Problem Setting}\label{section:modelsetting}
The goal of stochastic simulation is usually to estimate the system performance $H := E_{\xi\sim F^c}[h(\xi)]$, where $\xi$ is a random vector (r.v.) following the input distribution $F^c$, and $h$ is a function that can be evaluated by simulation. In this paper, we assume that the input distribution lives in a parametric family of distributions, that is, $F^c$ takes the parametric form $F(\cdot;\theta^c), \theta^c \in \Theta\subseteq \RR^s$, where $\theta^c$ is the true input parameter,  and $\Theta$ is the parameter space. While the value $\theta^c$ is unknown, we receive streaming data over time. Specifically, at each time $t$ ($t=1,2,\ldots$), we observe a new data point $\xi_t\sim F(\cdot;\theta^c)$ that is independent of the past data. With each new data point, we want to update the input model $F(\cdot;\theta_t)$ and quantify the impact of $F(\cdot;\theta_t)$ on the system performance estimation, particularly in a real-time fashion. { For $\forall~\theta\in\Theta,$ we further denote the probability density function of $F(\cdot;\theta)$ as $p(\cdot;\theta).$ We will use $\xi\sim p(\cdot;\theta)$ in the rest of the paper to denote drawing a sample $\xi$ from the input model $F(\cdot;\theta)$.}

We take a Bayesian approach to process data sequentially in time. The unknown input parameter is treated as a r.v. $\theta$ defined on $(\Theta,\cB_\theta,\pi_0)$, where $\pi_0$  is the prior distribution which is specified at time $t=0$.  At time stage $t$, the posterior distribution on $\theta$ has the probability density function (p.d.f.) $\pi_t := p(\theta|\xi_1,\ldots,\xi_t)$ and cumulative distribution function (c.d.f.) $\Pi_t$. Then at time stage $t+1$ when a new data point $\xi_{t+1}$ comes in,  the posterior distribution is updated according to
\begin{align}\label{bayes}
\pi_{t+1}(\theta):= p(\theta|\xi_1,\ldots,\xi_{t+1})= \frac{\pi_t(\theta)p(\xi_{t+1}|\theta)}{\int\pi_t(\theta)p(\xi_{t+1}|\theta)d\theta}. 
\end{align} We assume that we can sample  from $\pi_t$ for $t>0.$

To study the impact of input model on performance estimation, we introduce the notation $H(\theta) := E[h(\xi_{\theta})]$, which is a random variable induced by the r.v. $\theta$. The c.d.f. of the induced posterior distribution on $H(\theta)$ at time $t$ is then defined as
$$
G_t(h) := \PP(H(\theta)\leq h|\xi_1,\ldots,\xi_{t}).
$$
As in \cite{xie2014bayesian}, we use  the credible interval (also called the Bayesian confidence interval) of the induced posterior distribution on the performance measure to quantify input uncertainty. To be specific, the $(1-\alpha)100\%$ credible interval $[q_t, Q_t]$ is defined as
$$
G_t(Q_t) - G_t(q_t) = 1-\alpha,
$$
where $Q_t=\inf\{h|G_t(h)\geq1-\alpha/2\}$ and $q_t=\inf\{h|G_t(h)\geq \alpha/2\}.$
Note $G_t(\cdot)$ is the posterior distribution on the system performance. Intuitively, it represents our belief about the system performance based on the current dataset. If $H(\theta')$ can be evaluated exactly for each $\theta'\in\Theta$, the above credible interval can quantify the uncertainty  solely due to the input data. However, since we only have noisy evaluations of $H$, we need to estimate the credible intervals, which boils down to estimating the quantiles of $G_t(\cdot)$. Therefore, our goal is to estimate the quantiles of $G_t(\cdot)$ in a real-time manner at each time $t$ given the new data point $\xi_t$.

\section{Algorithms}
In this section, we present algorithms for quantifying input uncertainty in an online setting. We first present a Direct Monte Carlo  method, which serves as a benchmark and reveals the challenges of the online quantification problem. To address these challenges, we develop a novel framework of Two-Layer Importance Sampling. Under this framework, we then design two algorithms respectively for the following two scenarios: 1) all simulations are done at the beginning of the time, and no new simulation is allowed at any subsequent time stage;  and 2) a small number of simulation replications can be done at each time stage. In the end, we show that our framework also generalizes some other algorithms, including a Simple Importance Sampling algorithm and an online application of the Green Simulation method that was originally proposed in \cite{feng:2015WSCgreensim}  and \cite{Feng2017GreenSR}.

\subsection{Direct Monte Carlo Method}\label{Direct_MC}
Before introducing our algorithm, we first present a benchmark, Direct Monte Carlo method, to show the structure and challenges of this online quantification problem. Recall from last section that our goal is to estimate the quantiles of the performance posterior distribution $G_t$.  {The Direct  Monte Carlo method uses two-layer nested simulation:\begin{itemize}
	\item At the outer-layer, we draw $M$  $\theta-$samples $\{\theta^i_t\}_{i=1}^M$ from $\pi_t.$ The empirical distribution
	$$\overline{\pi}_t(\theta)=\frac{1}{M}\sum_{i=1}^M\cI(\theta = \theta_{t}^i)$$ is a consistent estimator for $\pi_t,$ by the famous  Glivenko-Cantelli Theorem. 	
	\item At the inner-layer, for each $\theta^i_t$ we carry out a finite number of simulation replications to obtain outputs $h(\xi^{i,j}_t), j=1,\ldots, N$, and then use the sample average $$\overline{H}^N(\theta_t^i) = \frac{1}{N}\sum_{i=1}^N h(\xi^{i,j}_t)$$ to estimate the true system performance under $\theta^i_t$. 
\end{itemize}

We then sort the performance estimate $\{\overline{H}^N(\theta_t^i)\}_{i=1}^M$ in an ascending order and find the quantile estimate.
The Direct Monte Carlo method is simple and easy to  implement. However, it requires a large number of simulation replications to obtain a sufficiently accurate estimate at any time stage. Specifically, at the outer-layer, we  need to choose a sufficiently large $M$ so that $\overline\pi_t$ is close to $\pi_t.$ At the inner-layer, we prefer choosing a large $N$ to control the simulation error.  This  high demand on simulation is generally not feasible in the online setting, since simulation is usually time consuming.
}

\subsection{Two-Layer Importance Sampling}
Our main idea is to adapt the well-known importance sampling (also called likelihood ratio) technique to enlarge the effective size of $\theta$-samples and the number of simulation replications that are used to estimate the system performance at any time stage. Specifically, at each time we apply IS to both the outer and inner layers in the following way.
\begin{itemize}
	\item At the outer-layer, we use importance sampling over multiple time stages to transform sets of $\theta$-samples from previous time stages to a weighted set of $\theta$-samples following the current posterior distribution $\pi_t$.
	\item At the inner-layer, we use importance sampling across different $\theta$-samples at the same time stage. That is, we use all the simulation outputs obtained under different $\theta$-samples to estimate the system performance under one target $\theta$-sample. We term this technique as Cross Importance Sampling (CIS).
\end{itemize}

Following the main idea outlined above, we develop the Two-Layer Importance Sampling (TLIS) algorithm in detail.
For any fixed $\theta\in\Theta$ and integer $K\in[1,t]$,  the posterior distribution $G_t$ can be rewritten as follows.
\begin{align}\label{eq_is}
G_t(h)&=\EE_{\theta_t\sim\pi_t}[\cI({h \geq H(\theta_t)})] \nonumber \\
&=\frac{1}{K}\sum_{k=0}^{K-1}\EE_{\theta_{t-k}\sim\pi_{t-k}}\left[\frac{\pi_t(\theta_{{t-k}})}{\pi_{{t-k}}(\theta_{{t-k}})}\cI\left({h\geq H(\theta_{t-k})}\right)\right]. 
\end{align}
Given a set of i.i.d. samples $\{\theta_{t-k}^i\}_{i=1}^M\widesim{\text{i.i.d.}}\pi_{t-k}(\theta)$, the expectation term in (\ref{eq_is}) can be estimated as follows:
$$
\EE_{\theta_{t-k}\sim\pi_{t-k}}\left[\frac{\pi_t(\theta_{{t-k}})}{\pi_{{t-k}}(\theta_{{t-k}})}\cI\left({h\geq H(\theta_{t-k})}\right)\right] \approx \frac{1}{M}\sum_{i=1}^M w_{t|t-k}^i \cI\left(h \geq H(\theta_{t-k}^i)\right),
$$
where $w_{t|t-k}^i$ is the likelihood ratio between $\pi_t$ and $\pi_{t-k}$ evaluated at $\theta_{t-k}^i$, i.e., 
\begin{align}\label{out_weight}
w_{t|t-k}^i=\frac{\pi_t(\theta_{t-k}^i)}{\pi_{t-k}(\theta_{t-k}^i)}\propto p(\xi_{t-k+1},...,\xi_{t}|\theta_{t-k}^i)=\prod_{\tau=t-k+1}^{t}p(\xi_{\tau}|\theta_{t-k}^i), 
\end{align}
where $p(\xi_{t-k+1},...,\xi_{t}|\theta_{t-k}^i)$ is the joint p.d.f. of $\{\xi_{t-k+1},...,\xi_{t}\}$ given that the input parameter is $\theta_{t-k}^i$, and the last equality follows from the independence among the data sequence $\{\xi_{\tau}\}_{\tau=1}^\infty.$  Therefore,  an unbiased c.d.f. estimate of $G_t$ using the $\theta$-samples from the most recent $K$ time stages is as follows:
\begin{align}\label{pi_estimate}
\hat{G}_{t}^{M,K}(\theta)=\frac{1}{KM}\sum_{k=0}^{K-1} \sum_{i=1}^M w_{t|t-k}^i \cI\left(h \geq H(\theta_{t-k}^i)\right).
\end{align}

In (\ref{pi_estimate}), the performance $H(\theta_{t-k}^i)$ cannot be evaluated exactly and has to be estimated through simulation. To make full use of all the simulation outputs, we propose Cross Importance Sampling, which applies importance sampling to the simulation outputs under different $\theta$'s in order to estimate the performance under one target $\theta$. More specifically, suppose we have  a set of $\theta$-samples $\{\theta^i\}_{i=1}^M$ and their corresponding $N$ simulation outputs $\{h(\xi^{i,j})\}_{j=1}^N,$ where $\xi^{i,j}\widesim{\text{i.i.d.}}p(\cdot|\theta^i), $ $i=1,...,M.$ Let's call $\{\theta^i\}_{i=1}^M$ the proposal parameter set. Our target is to estimate $H(\theta)$, where the target parameter $\theta$ could be inside or outside the proposal set  $\{\theta^i\}_{i=1}^M$.   Note that
\begin{align}\label{eq_cis}
H(\theta) &= \EE_{\xi\sim p(\cdot|\theta)}[h(\xi)] \nonumber \\ 
&= \frac{1}{M}\sum_{i=1}^M\EE_{\xi\sim p(\cdot|\theta^i)}\left[h(\xi)\frac{p(\xi|\theta)}{p(\xi|\theta^i)}\right].
\end{align}
Replacing the expectation terms in (\ref{eq_cis}) by the sample averages of $\left\{h(\xi^{i,j})\frac{p(\xi^{i,j}|\theta)}{p(\xi^{i,j}|\theta^i)}\right\}_{i,j}$,  $H(\theta)$ can by estimated by
\begin{align}\label{cis}
\hat{H}^{M,N}(\theta)&=\frac{1}{M} \sum_{i=1}^M\left\{\frac{1}{N}\sum_{j=1}^N h(\xi^{i,j})\frac{p(\xi^{i,j}|\theta)}{p(\xi^{i,j}|\theta^i)}\right\}=\frac{1}{NM} \sum_{i=1}^M\sum_{j=1}^N \nu^{i,j}(\theta)h(\xi^{i,j}),
\end{align}
where $\nu^{i,j}(\theta)=\frac{p(\xi^{i,j}|\theta)}{p(\xi^{i,j}|\theta^i)}.$
\begin{figure}[htb!]
	\center
	\includegraphics[scale=0.1]{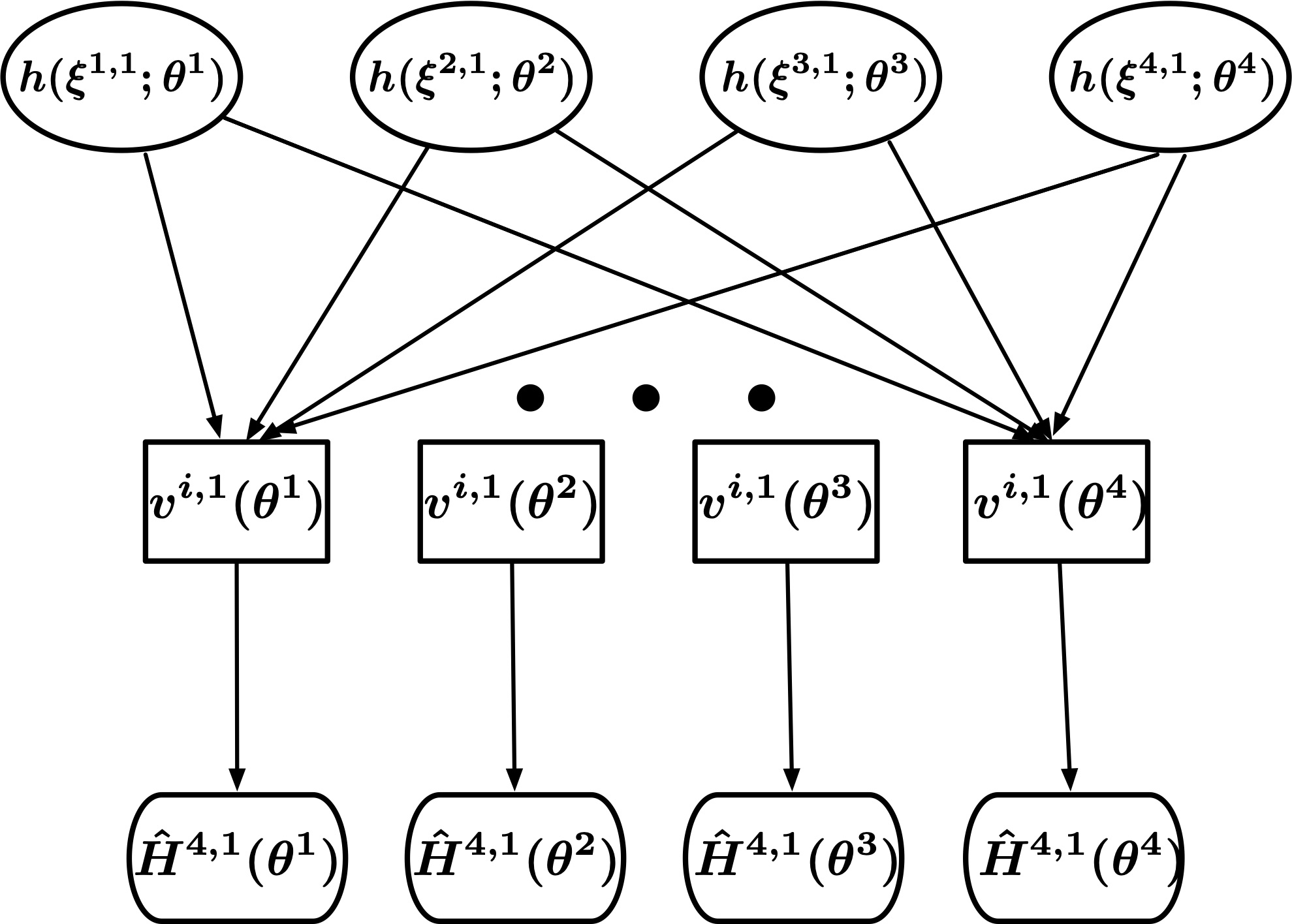}
	\caption{Illustration of Cross Importance Sampling. ($M=4,$ $N=1,$ same proposal and target parameter sets.)}
	\label{fig_cis}
\end{figure}
Compared to the sample average estimator obtained by the Direct Monte Carlo method, \eqref{cis} uses $NM$  instead of $N$ simulation outputs to estimate any single system performance. Thus, as we will show later, when certain condition is satisfied, the variance of the estimator can be substantially reduced. Note that one special case is that the target parameter set is exactly the proposal parameter set. In this case,  the system performance estimator  for each $\theta$-sample in this set  uses simulation outputs under all the samples  crosswise. That is why we call this method Cross Importance Sampling (CIS). An illustration of  CIS is shown in Figure \ref{fig_cis}, when there are four $\theta$-samples and one simulation output under each $\theta$-sample.
Replacing $H(\theta_{t-k}^i)$ in (\ref{pi_estimate}) by  $\hat H^{M,N}(\theta_{t-k}^i)$, we finally obtain the empirical estimator of the c.d.f. $G_t(h),$ as follows.
\begin{align}\label{CDF_estimate}
\hat{G}_{t}^{M,N,K}(h)=\frac{1}{KM}\sum_{k=0}^{K-1} \sum_{i=1}^M w_{t|t-k}^i \cI\left(h \geq \hat H^{M,N}(\theta_{t-k}^i)\right).
\end{align}

\begin{remark}
	Cross Importance Sampling should not be used to estimate average system performance under a given input distribution. Since every system performance estimate uses all the simulation outputs, they are	positively correlated. When taking average, the mean estimator will suffer from extreme large variance, and thus CIS is not recommended. In the quantification of input uncertainty, however, we care about one single quantile. There is no issue of positive correlation here in quantile estimator.
\end{remark}

Our Two-Layer Importance Sampling algorithm is presented below in Algorithm \ref{General}.
In the step of Cross Importance Sampling, the proposal parameter set should be chosen appropriately for the practical scenario. We discuss two main scenarios driven by the pace of decision making compared with data arrivals: I) Fast decision making. For instance, in the stock market, price changes every second, and investors need to update decisions and quantify risks in real time. II) Moderate (but still online) decision making. This scenario requires up-to-date but not immediate decision making. For example, in inventory management decisions are made at a moderate pace, such as weekly or monthly. These two scenarios allow different amount of simulation experiments we can carry out between decision epochs. In the first scenario, there is hardly any time for new simulation experiments between decisions, and we can only carry out simulation experiments in specific time periods such as market closures. In the second scenario, however, we can simulate the system to estimate the performance after new data come in, but we can only afford a small number of simulation replications because each replication is expensive and there is often limited computational resource. Targeting at these two scenarios, we propose the corresponding CIS estimators by choosing appropriate proposal parameter sets.

\setcounter{algorithm}{0}
\begin{algorithm}[htb!]
	\setlength{\textfloatsep}{0pt}
	\DontPrintSemicolon
	\caption{ \it Two-Layer Importance Sampling (TLIS) }
	\label{General}
	\textbf{Input}: Data sequence $\{\xi_{t},\,t=1,2,...\}.$ \\
	\textbf{Output}: Estimator of $G_t$ and its quantile. \\ 
	\textbf{Initialization}: Specify a prior distribution $\pi_0$. Draw i.i.d. samples $\{\theta_0^1,...,\theta_0^M\}$ from $\pi_0$. For each $i=1,...,M$, run $N$ simulation replications at $\theta_0^i$ to obtain the outputs $h(\xi^{i,j}_0)$ with $\xi^{i,j}_0 \widesim{\mathrm{i.i.d.}}p(\cdot|\theta_0^i), j=1,...,N.$\\
	\textbf{ At Time stage $t ~(t \geq 1)$}: A new data point $\xi_{t}$ arrives. The following steps are carried out.
	\begin{itemize}
		\item[1.] \textit{Outer-layer Importance Sampling}: For each $\theta_{t-k}^i$, calculate its importance weights $w_{t|t-k}^i$,$k=1,...,\min\{t,K\},$ $i=1,...,M,$ according to \eqref{out_weight}.
		\item[2.] {\textit{Cross Importance Sampling}: Compute $\pi_t$, and draw samples $\theta_t^i\widesim{\text{i.i.d.}}\pi_t, i=1,...,M.$ Choose a proposal parameter set $\{\theta^i\}_{i=1}^M$ according to scenarios.\\
		$\bullet$ Scenario I: $\{\theta^i\}_{i=1}^M=\{\theta_0^i\}_{i=1}^M.$ Then estimate $H(\theta_t^i)$ by $\hat H^{M,N}(\theta_t^i)$ according to \eqref{CIS1}.\\
		$\bullet$ Scenario II: $\{\theta^i\}_{i=1}^M=\{\theta_t^i\}_{i=1}^M.$ For each $\theta^i$, run simulations to obtain the  simulation outputs  $\{h(\xi^{i,j})\}_{j=1}^N,$ where $\xi^{i,j}\widesim{\text{i.i.d.}}F(\cdot;\theta^i)$, and estimate $H(\theta_t^i)$ by $\hat H^{M,N}(\theta_t^i)$ according to \eqref{CIS2}.
			}
		\item[3.] \textit{Quantification}: Compute the c.d.f. estimator $\hat{G}_{t}^{M,N,K}(h)$ according to \eqref{CDF_estimate}, and get its $\alpha$-quantile $~~~~~~\hat{q}_t^{\alpha}= \inf\left\{h: \hat{G}_t^{M,N,K} (h)\geq \alpha \right\}$.		
	\end{itemize}		
\end{algorithm}

\noindent $\bullet$ {\bf Scenario I}: As we mentioned before, in this scenario we can only carry simulations in some special time periods. Here, we assume this time period is $t=0.$ In later stages, we are unable to carry out any new simulation experiments and can only use the simulation outputs at $t=0$ to estimate the system performance. To this end, in each time stage, when we apply Cross Importance Sampling, the proposal parameter set $\{\theta^i\}_{i=1}^M$ will be chosen as $\{\theta_{0}^i\}_{i=1}^M.$ Then we do not need any simulation experiments but can still achieve good estimates of the system performance of the corresponding new outer-layer input parameter samples.  The CIS estimator \eqref{cis} for scenario I is shown as follows.
\begin{align}\label{CIS1}
\hat{H}^{M,N}(\theta_t^i)	=\frac{1}{NM}\sum_{l=1}^M\sum_{j=1}^N\frac{p(\xi_0 ^{l,j}|\theta_t^i)}{p(\xi_0^{l,j}|\theta_0^l)}h(\xi_0^{l,j}),~~\forall i=1,...,M.
\end{align}
Note that our theoretical analysis will justify that this algorithm can achieve relatively accurate quantification for a long time horizon compared with naive Monte Carlo method. However, when new simulation runs are allowed (e.g., during market closure),  we recommend to restart the algorithm by discarding the current simulation outputs and obtaining new outputs under a newly drawn set of $\theta$-samples from the latest posterior distribution. 

\noindent $\bullet$ {\bf Scenario II}: In this scenario, we can do a small number of simulation replications to estimate the system performance. For each newly drawn i.i.d. samples $\theta_t^i,$ $i=1,...,M,$ we do $N$ simulations and obtain the system performance $\{h(\xi_t^{i,j})\}_{j=1}^N,$  $i=1,...,M.$ We then apply CIS with the proposal parameter set chosen as $\{\theta_{t}^i\}_{i=1}^M$.  The CIS estimator  \eqref{cis} for scenario II is shown as follows.
\begin{align}\label{CIS2}
\hat{H}^{M,N}(\theta_t^i)	=\frac{1}{NM}\sum_{l=1}^M\sum_{j=1}^N\frac{p(\xi_t ^{l,j}|\theta_t^i)}{p(\xi_t^{l,j}|\theta_t^l)}h(\xi_t^{l,j}),~~\forall i=1,...,M.
\end{align}
 For notational simplicity, in the rest of the paper, we call Algorithm \ref{General} using CIS estimators \eqref{CIS1} and \eqref{CIS2} respectively as Algorithm TLIS-1 and Algorithm TLIS-2.

\subsection{Other Algorithms}
Some other algorithms can also be interpreted from our Two-Layer Importance Sampling framework. These includes the Direct Monte Carlo  method mentioned in Section \ref{section:modelsetting}, a Simple Importance Sampling method, and an online application of the Green Simulation algorithm that was originally proposed in \cite{feng:2015WSCgreensim}  and \cite{Feng2017GreenSR}.

\noindent $\bullet$ {\bf Direct Monte Carlo Method}: When $K=1$ and the proposal parameter set  includes the target parameter sample only, the Two-Layer Importance Sampling is reduced to the Direct Monte Carlo method. The c.d.f. estimator of $G_t$ obtained in each time stage can be written as follows.
$$\overline{G}_t^{M,N}(h)=\frac{1}{M}\sum_{i=1}^M  \cI\left(h \geq \overline{H}^N(\theta_t^i)\right),$$
where $\theta_t^i\widesim{\text{i.i.d.}}\pi_t,$ and $\overline{H}^N(\theta_t^i)=\frac{1}{N}\sum_{i=1}^N h(\xi^{i,j}_t)$ is the sample average estimate of the system performance of $H(\theta_t^i),$ where
$\{\xi^{i,j}_t\}_{j=1}^N\widesim{\text{i.i.d.}}p(\cdot|\theta_t^i).$ Note that the Direct Monte Carlo method is hardly applicable to the online setting, due to its high computational cost at each time $t$.

\noindent $\bullet$ {\bf Simple Importance Sampling}: Unlike the Direct Monte Carlo method, the Simple Importance Sampling algorithm supports fast decision making in Scenario I. It only draws new input parameter samples and runs simulation experiments in the beginning of the algorithm. When a new data point comes, it simply transforms these samples and the corresponding system performance estimates by importance sampling to the target new posterior distribution. The c.d.f. estimator of $G_t$ obtained at each time stages can be written as follows.
$$\tilde{G}_t^{M,N}(h)=\frac{1}{M}\sum_{i=1}^M \frac{\pi_t(\theta_0^i)}{\pi_0(\theta_0^i)}\cI\left(h\geq \tilde{H}^N(\theta_0^i)\right),$$
where $\theta_0^i\widesim{\text{i.i.d.}}\pi_0,$ and $\tilde{H}^N(\theta_0^i)=\frac{1}{N}\sum_{j=1}^N h(\xi^{i,j}_0)$ is the sample average estimate of the system performance of $H(\theta_0^i),$ where
$\{\xi^{i,j}_0\}_{j=1}^N\widesim{\text{i.i.d.}}p(\cdot|\theta_0^i).$ { One significant difference between Simple Importance Sampling and TLIS-1 is that Simple Importance Sampling never draws new $\theta-$samples from current posterior distribution and thus, heavily depends on the $\theta-$samples at the initialization stage.  As the posterior distribution evolves, $\pi_t$ could significantly differ from $\pi_0$ and as a result, the variance of the importance weights can explode over time.}

\noindent $\bullet$ {\bf Green Simulation}: \cite{feng:2015WSCgreensim}  and \cite{Feng2017GreenSR} design the Green Simulation algorithm to save the simulation budget for off-line system performance estimation. It  can also be applied for the online quantification of input uncertainty. The main difference of Green Simulation from Two-Layer Importance Sampling is that it  does not include the inner-layer Cross Importance Sampling, but uses the sample average estimator instead. Green Simulation also requires new simulations at every time stage, so it is only applicable to Scenario II mentioned above. Without new simulations at each time stage, Green Simulation reduces to the Simple Importance Sampling method. The c.d.f. estimator of $G_t$ obtained at each time stage can be written as follows.

$$\breve{G}_{t}^{M,N,K}(h)=\frac{1}{KM}\sum_{k=0}^{K-1} \sum_{i=1}^M \frac{\pi_t(\theta_{t-k}^i)}{\pi_(\theta_{t-k}^i)} \cI\left(h \geq \breve H^{N}(\theta_{t-k}^i)\right),$$
where $\theta_s^i\widesim{\text{i.i.d.}} \pi_s,~s=t-k,...,t,$   and $\breve{H}^N(\theta_{t-k}^i)=\frac{1}{N}\sum_{j=1}^N h(\xi^{i,j}_{t-k})$ is the sample average estimate of the system performance of $H(\theta_{t-k} ^i),$ where
$\{\xi^{i,j}_{t-k}\}_{j=1}^N\widesim{\text{i.i.d.}}p(\cdot|\theta_{t-k}^i).$
We will empirically compare Green Simulation algorithm with our Two-Layer Importance Sampling algorithm in Experiment 2 in Section \ref{experiment}.

\section{Convergence Analysis}
In this section, we analyze the convergence and the convergence rate of Two-Layer Importance Sampling. We consider the case where the parametric input distribution belongs to the exponential families of distributions. Note that since every distribution from exponential families has a conjugate prior, they are widely used in parametric models.  We will first briefly introduce the exponential families of distributions and show some important properties of them. Using these properties, we will then show that Two-Layer Importance Sampling can achieve consistent estimators and faster convergence speed than the Direct Monte Carlo Method. We will further discuss in Appendix \ref{discuss_general} the general conditions for our  convergence results when applying the algorithm to an arbitrary parametric input distribution.

Since there are multiple sources of randomness coming from data, input distribution, and simulation, we first  construct the probability space required.  Suppose $\theta$ takes value in a parameter space $\Theta\subset\RR^s,$ which is equipped with a Borel $\sigma-$ algebra $\cB_{\theta}$.  Let $(\Omega, \cF,\PP^1_\theta)$ be the probability space in which $\xi_\theta$ takes value, where $\Omega\in \RR^d$, $\cF$ is the Borel $\sigma$-algebra on $\Omega,$ and $\PP^1_\theta$ is the distribution of $\xi_\theta$ given $\theta.$ Then $(\Omega^n, \cF^n,\PP_\theta^n)$ is the probability space in which $n$ i.i.d. copies of $\xi_\theta$ take values. Here, $\cF^n$ is the product $\sigma$-algebra $\cF\otimes\cF\otimes\cdot\cdot\cdot\otimes \cF$, and $\PP_\theta^n$ is the product measure $\PP^1_\theta\times\PP^1_\theta\cdot\cdot\cdot\times \PP^1_\theta.$ By Kolmogorov's extension theorem (see, e.g., \cite{durrett2019probability}, Theorem A.3.1), we can extend $(\Omega^n, \cF^n,\PP_\theta^n)$ to  $(\Omega^\NN, \cF^\NN,\PP_\theta^\NN),$ where $\Omega^\NN$ is the space of all infinite sequences in $\Omega,$ and $\cF^\NN$ is the $\sigma$-algebra generated  by all sets taking the following form
$$\left\{\tilde \xi\in \Omega^\NN:(\tilde \xi_1,\tilde \xi_2,\ldots,\tilde \xi_n)\in \cF^n\right\},$$
where $\tilde \xi_i$ is the $i$th component of $\tilde \xi.$ The probability $\PP_\theta^\NN$ is defined as the product measure that
coincides with $\PP_\theta^n$ on $ \cF^n$, i.e.,
$$\PP_\theta^\NN\left( \left\{\tilde \xi\in \Omega^\NN:(\tilde \xi_1,\tilde \xi_2,\ldots,\tilde \xi_n)\in H\right\}\right)=\PP_\theta^n(H),~~\forall H\in \cF^n.$$
We remark that this construction follows that in \cite{wu2018bayesian}. Please refer to Section 2.1 in \cite{wu2018bayesian} for more details.

Moreover, we  make the following assumption on the parameter space which holds throughout the rest of the paper.
\begin{assumption}\label{TLIS-1}
	The parameter space $\Theta$ is compact.
\end{assumption}
This assumption can be easily satisfied in practice. For example, we can use our prior knowledge on the parameter to set up such a compact set.

\subsection{ Exponential Families of Distributions (EFDs) }\label{sec_exp}

Exponential families of distributions (EFDs) include most of the commonly used distributions, such as Gaussian distributions, exponential distributions, and Poisson distributions.  One important property EFDs enjoy is the existence of  conjugate priors, which makes them popular in Bayesian statistics and hence suitable for the role of input models in our setting. EFDs have the following form of probability density functions: 
\begin{align*}
p(x\mid\theta) = \kappa(x) \exp \left ( {\theta}^\top{T}(x) -A(\theta)\right ),
\end{align*}
where ${T}(x), ~A(\theta)$, and $ \kappa(x)$ are known functions. Denote the true parameter as $\theta^c\in\Theta\subseteq\cN\subseteq \RR^s,$ where $\cN$ is the natural parameter space defined as
$$\cN=\left\{\theta: A(\theta)<\infty \right\}.$$ 

Suppose at time $t,$ we have a data point $\xi_t\sim p(\cdot|\theta^c).$
Assuming  non-informative uniform prior $\pi_0$, i.e.,
$\pi_0(\theta)={\cI}(\theta\in\Theta)/\mu(\Theta),$ where $\mu$ is the Lebesgue measure on $\RR^s,$ then the posterior distribution $\pi_t$ of $\theta$  given $\{\xi_i\}_{i=1}^t$ is as follows.
$$\pi_t(\theta) = f(\boldsymbol\chi_t,t)\pi_0( \theta) \exp \left (\theta^\top \boldsymbol\chi_t - t A(\theta) \right ),$$
where $$ \boldsymbol\chi_t=\sum_{i=1}^t {T}(\xi_i),~\text{and}~f(\boldsymbol\chi_t,t)=\int_\Theta \pi_0(\theta) \exp \left (\theta^\top \boldsymbol\chi_t - t A(\theta) \right )<\infty.$$

When we apply Two-Layer Importance Sampling, we need to reuse the $\theta-$samples from the posterior distributions at previous time stages. Recall that the estimator in \eqref{pi_estimate}, i.e., \begin{align*}
\hat{G}_{t}^{M,K}(\theta)=\frac{1}{KM}\sum_{k=0}^{K-1} \sum_{i=1}^M \frac{\pi_t(\theta_{t-k}^i)}{\pi_{t-k}(\theta_{t-k}^i)} \cI(h \geq H(\theta_{t-k}^i)),
\end{align*}
is unbiased and has variance that depends on the likelihood ratio $ \pi_t(\theta_{t-k}^i)/\pi_{t-k}(\theta_{t-k}^i), ~k=0,...,K-1$. When the variance of $\pi_t/\pi_{t-k}$ is extremely large, $\pi_{t-k}$ is not a good proposal distribution for $\pi_t,$ and we should not reuse samples from time stage $t-k.$ Fortunately, the next lemma verifies that for EFDs and a fixed $K,$ the variance of the likelihood ratio is bounded almost surely (a.s.).  {Due to space limit, we defer all the technical proofs   appeared in Section \ref{sec_exp} to Appendix \ref{exp_proof}.}
\begin{theorem}\label{thm:exp_out_bounded}
	Suppose $\nabla^2A(\theta)$ is positive definite for all $\theta\in\cN,$ and  $\norm{{T}(\cdot)}_2$ is bounded in $\Omega.$ For any fixed constant $k\geq 0,$ as $t\rightarrow\infty,$
	\begin{align}\label{exp_weight_var}
	\EE\left(\frac{\pi_{t}(\theta)}{\pi_{\max\{t-k,0\}}(\theta) }\right)^2\rightarrow 1 ~~\text{ a.s. ($\PP^\NN_{\theta^c}$)},
	\end{align}
	and
	$$\EE\left(\frac{\pi_{t}(\theta)}{\pi_{\max\{t-k,0\}}(\theta) }\right)^3<\infty ~~\text{ a.s. ($\PP^\NN_{\theta^c}$)}.$$
\end{theorem}

Note that \eqref{exp_weight_var} only shows that the variance of $\pi_{t}/\pi_{t-k}$ is bounded  for any single $k$ almost surely. In the estimator \eqref{pi_estimate}, we reuse $K$ time stages and need all the likelihood ratios to be jointly bounded. This can be verified by the following corollary.
\begin{corollary}\label{lem_bound}
	Suppose $\nabla^2A(\theta)$ is positive definite for all $\theta\in\cN,$ and  $\norm{{T}(\cdot)}_2$ is bounded in $\Omega.$ For any fixed constant $K\geq 0,$ 
	$$\PP^\NN_{\theta^c}\left(\exists C>0,\EE\left(\frac{\pi_{t}(\theta)}{\pi_{ \max\{t-k,0\}}(\theta) }\right)^2<C,\forall 1\leq k\leq K, \forall t>0.\right)=1.$$ 
\end{corollary}

Corollary \ref{lem_bound} shows that the variance of the weight is uniformly bounded in the sense that for any given $t$ and any $1\leq k\leq K$, the variance of the importance ratio $\frac{\pi_{t}(\theta)}{\pi_{ \max\{t-k,0\}}(\theta)}$ is bounded almost surely. Thus, importance sampling can be used without the worry about  variance explosion.


In addition to the outer-layer importance samling, we also apply CIS in the inner-layer to improve the system performance estimates. The next theorem shows that under certain conditions, the CIS estimator also has bounded variance.
\begin{theorem}\label{thm:exp_in_bound}
	Suppose $\Theta'=\{2\theta_1-\theta_2\big|\theta_1\in\Theta,\theta_2\in\Theta \}\in\cN$ and  for every $\theta\in\Theta',$ $\EE\{h(\xi_{\theta})^2\}<\infty.$ There exists some constant $C_1>0$ such that 
	$$\sup_{\theta_1,\theta_2\in \Theta}\Var\left\{\frac{p(\xi|\theta_1)}{p(\xi|\theta_2)}h(\xi)\Big| \theta_1,\theta_2\right\}\leq C_1.$$
\end{theorem}
For EFDs  whose  natural parameter space is $\RR^s$ (e.g., normal distributions with known variance), the assumption  $\Theta'=\{2\theta_1-\theta_2\big|\theta_1\in\Theta,\theta_2\in\Theta \}\in\cN$ is naturally satisfied. For other distributions, this assumption can be violated. However, in practice, we only need  a weaker assumption: For any $\theta_i$ and $\theta_j$ in the sample set $\{\theta_1,...,
\theta_N\},$ $2\theta_i-\theta_j\in \cN.$ In fact, this weaker assumption holds with high probability when $t$ is large. To see it clearly, we consider the natural parameter space to be $R^+=(0,+\infty).$ By Bernstein-von Mises theorem \citep{van2000asymptotic}, we know that as $t\rightarrow \infty$, $$\EE_{\theta\sim\pi_t}[\theta]\rightarrow \theta^c~~\text{and}~~\Var_{\theta\sim\pi_t}[\theta]=O\left(\frac{1}{t}\right).$$
By Chebyshev's inequality,  we know that the sample will concentrate in the neighborhood of the expectation with high probability,
$$\PP(|{\theta_t-\EE\theta_t}|\leq {\EE\theta_t}/3)\geq 1-9\frac{\Var{\theta_t}}{(\EE\theta_t)^2}=1-O\left(\frac{1}{t}\right).$$
Given the inner-layer sample size $N$, the probability that all $N$ $\theta-$samples fall in the region  $\left(\frac{2{\EE\theta_t}}{3},\frac{4{\EE\theta_t}}{3}\right)$ is approximately $1-O\left(\frac{N}{t}\right)$ when $t$ is large. Moreover, $ \forall \theta_1, \theta_2\in \left(\frac{2}{3}{\EE\theta_t},\frac{4}{3}{\EE\theta_t}\right),$  $2\theta_1-\theta_2>0$  holds. This implies the weaker assumption is satisfied with very high probability. However,  the weaker condition could still be violated when $t$ is small. Thus, CIS is not suitable for the cases where the posterior distribution is dispersive due to the lack of input data when $t$ is small. We can empirically observe this phenomenon in our numerical example in Section \ref{experiment}.


\subsection{Convergence Results}
Theorem \ref{thm:exp_out_bounded}, Corollary \ref{lem_bound}, and Theorem~\ref{thm:exp_in_bound} together provide us the key properties of EDFs to show the convergence property of  Two-Layer Importance Sampling algorithm. We first re-state the following assumption such that these properties hold for EDFs.
\begin{assumption}\label{assumption_Exp}
	$~$
	\begin{itemize}
		\item[1.] Condition in Theorem \ref{thm:exp_out_bounded}. That is, suppose $\nabla^2A(\theta)$ is positive definite for all $\theta\in\cN$, and  $\norm{{T}(\cdot)}_2$ is bounded in $\Omega.$
		\item[2.] Condition in Theorem \ref{thm:exp_in_bound}. That is, $\Theta'=\{2\theta_1-\theta_2\big|\theta_1\in\Theta,\theta_2\in\Theta \}\in\cN$ and  for every $\theta\in\Theta',$ $\EE\{h(\xi_{\theta})^2\}<\infty.$
	\end{itemize}
\end{assumption}
We introduce the following notations. Recall that $N,$ $M,$ and $K$ are the number of inner-layer simulation replications, number of outer-layer $\theta-$samples, and number of previous time stages reused, respectively. Recall that the estimator of the posterior distribution of the system performance at time $t$ under the hyper-parameter tuple, $(M,N,K),$ defined in \eqref{CDF_estimate}, is denoted as 
$\hat{G}_{t}^{M,N,K}(h).$
If we ignore the inner layer simulation error, i.e., $N=\infty,$ the c.d.f. estimate defined in \eqref{pi_estimate}  is denoted as 
$\hat{G}_{t}^{M,K}(h).$ 
Note that since we use the true system performance $H(\theta),$ there is only input uncertainty in $\hat{G}_{t}^{M,K}(h).$ Denote the p.d.f of $G_t$ as $g_t.$ We define the estimators of $g_t$  corresponding to $\hat{G}_{t}^{M,N,K}(h)$ and $\hat{G}_{t}^{M,k}(h)$ respectively, as follows.
$$\hat{g}_{t}^{M,N,K}(h)=\frac{1}{KM}\sum_{k=0}^{K-1} \sum_{j=1}^M w_{t|t-k}^j \cI\left(h = \hat H^{M,N}(\theta_{t-k}^i)\right),$$
$$\hat{g}_{t}^{M,K}(h)=\frac{1}{KM}\sum_{k=0}^{K-1} \sum_{j=1}^M w_{t|t-k}^j \cI\left(h = H(\theta_{t-k}^i)\right).$$ 
For any fixed $\alpha\in(0,1),$ let $\hat{q}_{t}^{M,N,k}$,  $\hat{q}_{t}^{M,k}$, and $q_{t}$ denote the quantile estimators obtained by TLIS with and  without simulation uncertainty, and the true $\alpha $ quantile at time $t$, i.e.,
\begin{align*}
&\hat{q}_{t}^{M,N,K}=\inf\left\{h: \hat{G}_t^{M,N,K}(h)\geq \alpha \right\},\\
&\hat{q}_{t}^{M,K}=\inf\left\{h: \hat{G}_t^{M,K}(h)\geq \alpha \right\},\\
&q_{t}=\inf\left\{h: G_t(h)\geq \alpha \right\}.
\end{align*}
Since we assume $\alpha$ is fixed, here we omit $\alpha$  for notational simplicity.
Recall that the estimate of the posterior distribution of the input parameter  is as follows.
$$\hat{\pi}_{t}^{M,K}(h)=\frac{1}{KM}\sum_{k=0}^{K-1} \sum_{i=1}^M w_{t|t-k}^j \cI\left(\theta = \theta_{t-k}^i\right).$$

We next analyze  the asymptotic properties of  the quantile estimator $\hat{q}_t^{M,N,K}$, as the inner and outer sample sizes ($N$ and $M$) both go to infinity.
Specifically, we prove the consistency and asymptotic normality under the following set of conditions.

\subsubsection{Consistency}
We first show that our proposed estimator $\hat{q}_{t}^{M,N,K} $ is consistent, which implies when we run enough simulation replications, we can get a precise quantification of the input uncertainty. It turns out that, under Assumption \ref{assumption_Exp},  $\hat{q}_{t}^{M,N,K} $  is consistent in the sense that it converges to ${q}_{t}$ as $N$  first goes to infinity and then $M$ goes to infinity, which is shown in the following theorem. 

\begin{theorem}\label {thm:consistency}
	Under  Assumption \ref{assumption_Exp}, we have 
	$$\lim_{M\rightarrow\infty}\lim_{N\rightarrow\infty}\hat{q}_{t}^{M,N,K}= q_{t},~~ \mathrm{a.s.}$$
\end{theorem}
Due to space limit, we only provide a proof sketch here, please refer to Appendix \ref{proof:consistency} for the detailed proof of the theorem and technical lemmas.
\proof{Proof Sketch.}
	
	Note that the estimation error can be decomposed according to the source of uncertainty. Specifically, we have 
	\begin{align*}
	&\hat{q}_{t}^{M,N,K}-q_{t}=\underbrace{\hat{q}_{t}^{M,N,K}-\hat{q}_{t}^{M,K}}_{\text{Inner-Layer Error}}+\underbrace{\hat{q}_{t}^{M,K}-q_{t}.}_{\text{Outer-Layer Error}}
	\end{align*}
	Here, the inner-layer error is caused by the simulation uncertainty, while the outer-layer error comes from the input uncertainty. The following lemma shows that when $N$ first goes to infinity, the inner-layer error will vanish.
	\begin{lemma} \label {thm:in_consistency}
		Given $\{\theta_i\}_{i=1}^M,$ we have 
		$$\lim_{N\rightarrow\infty} \hat{q}_{t}^{M,N,K}=\hat{q}_{t}^{M,K},~~ \mathrm{a.s.}$$
	\end{lemma}
	Intuitively, the  number of inner-layer simulation replications $N$ going to infinity ensures that for any  $\theta$ in $\{\theta_i\}_{i=1}^M,$ $\hat H^{M,N}(\theta)\rightarrow H(\theta)$ almost surely by the law of large number. Suppose we have $M$ parameter samples $\theta_1,...,\theta_M.$ We sort $\{H(\theta_i)\}_{i=1}^{M}$ in the ascending order and get $H(\theta_{(1)})\leq H(\theta_{(2)})\leq\ldots \leq H(\theta_{(M)})$. Similarly, we sort the system performance  estimates $\{\hat H^{M,N}(\theta_i)\}_{i=1}^{M}$ in ascending order and get $\hat H^{M,N}(\theta^{(1)})\leq \hat H^{M,N}(\theta^{(2)})\leq\ldots \leq \hat H^{M,N}(\theta^{(M)})$. Note that for any $i=1,..,M,$ the order statistics $\theta_{(i)}$ does not necessarily equal to $\theta^{(i)}$ since the simulation uncertainty may change the order. Let $$\epsilon=\inf\left\{H(\theta_{(i+1)})-H(\theta_{(i)})\Big|H(\theta_{(i+1)})\neq H(\theta_{(i)})\right\}.$$ When $N$ is large enough, such that   $$|{H(\theta_i)-\hat H^{M,N}(\theta_i)}|\leq \frac{\epsilon}{3},~\forall i=1,..,M.$$ If we have $H(\theta_{(i)})<H(\theta_{(i+1)})$, then for any $1\leq i\leq M-1,$ $$\hat H^{M,N}(\theta_{(i)})\leq H(\theta_{(i)})+\frac{\epsilon}{3}\leq H(\theta_{(i+1)})-\epsilon+\frac{\epsilon}{3}<H(\theta_{(i+1)})-\frac{\epsilon}{3}\leq \hat H^{M,N}(\theta_{(i+1)}), $$
	which implies $\theta_{(i)}=\theta^{(i)}, i=1,..., M.$
	It follows that $(\theta^{(1)},\theta^{(2)},...,\theta^{(M)})\rightarrow (\theta_{(1)},\theta_{(2)},...,\theta_{(M)})$ almost surely as $N\rightarrow \infty.$ Then, one can show that the inner-layer error vanishes when $N$ is large enough.

Next, we show the consistency of the outer-layer by the following lemma.
	\begin{lemma}\label {thm:out_consistency}
		Under  Assumption \ref{assumption_Exp}, we have   $$\lim_{M\rightarrow\infty}\hat{q}_{t}^{M,K}= q_{t},~~ \mathrm{a.s.}$$
	\end{lemma}
	The proof of the above lemma is similar to that of Theorem 4.2 in \cite{egloff2010quantile}. We only need to verify that for any $\delta>0,$
	$$\PP\left(\hat{q}_{t}^{M,K}\leq q_{t}-\delta,i.o.\right)=0\text{~~ and ~~}\PP\left(\hat{q}_{t}^{M,K}\geq q_{t}+\delta,i.o.\right)=0.$$
	This can be done by Borel-Cantelli Lemma. Combining Lemma \ref{thm:in_consistency} and \ref{thm:out_consistency}, we prove the result.
\endproof

\subsection{Asymptotic Convergence Rates}
{ The consistency result shows that our proposed quantile estimator converges to the true quantile asymptotically. In this section, we study the convergence rate. }
The asymptotic convergence rate also helps  demonstrate the advantage of reusing previous outer-layer $\theta-$samples and applying inner-layer CIS over other methods. To see the improvement of each layer, we show the convergence rates for outer-layer and inner-layer separately. {Due to space limit, we defer all the technical proofs  to Appendix \ref{rate_proof}.}

Recall that at time $t,$ we reuse all the $\theta-$samples from the past $K$ stages. Thus, the estimator uses a total number of $KM$ $\theta-$samples. On the one hand, in the ideal case, from the central limit theorem we have a convergence rate of order $O(\frac{1}{\sqrt{KM}}),$ which improves the rate by a factor $O(\frac{1}{\sqrt{K}}).$ On the other, importance sampling may  change the variance of the estimator, which is determined by the importance weight.  Assumption \ref{assumption_Exp}.1 ensures the boundedness of the variance of the importance weight. Thus, we have the following theorem.
\begin{theorem} \label{thm:out_rate}
	Under Assumption \ref{assumption_Exp}, we have 
	$$(KM)^{\frac{1}{2}}\left(\hat{q}_{t}^{M,K}- q_{t}\right)\Rightarrow \sigma_t^{K} N(0,1),$$
	as $M\rightarrow\infty,$ where 
	$$\sigma_t^{K}=\sqrt{\frac{1}{K}\sum_{k=0}^{K-1}\EE\left\{  (w_{t|t-k}^j)^2 \cI\left( H(\theta_{t-k}^i)\leq q_t\right)-\alpha^2\right\}}\Bigg/g_t(q_t).$$
	{Moreover, if we further know that $g_t(q_t)$ is lower bounded by some positive constant for all $t>0,$ then there exists some constant $C_2>0,$ such that
		$$\sigma_t^{K}\leq C_2,~\forall t>0,K\geq1.$$}
\end{theorem}
	
{
	\begin{remark}
		~
		\begin{itemize}
			\item The boundedness of $\sigma_t$ requires $g_t(q_t)$ to be uniformly bounded from zero for all $t.$ For many widely used distributions in EFDs (e.g., normal distribution), this can be  verified  when $H(\theta)$ is strictly monotone and smooth in  $\Theta.$ For more details, please refer to Appendix \ref{bound_quantile_density}. Due to current technique limit, we cannot verify it for all EFDs.
			\item When $K=1$, we get the outer-layer convergence rate of the Direct Monte Carlo method, i.e.,
			$$M^{\frac{1}{2}}\left(\hat{q}_{t}^{M,1}- q_{t}\right)\Rightarrow \sigma_t^{1} N(0,1), ~\text{as}~M\rightarrow\infty,$$
			 where  $\sigma_t^{1}=\sqrt{\EE\left\{   \cI\left( H(\theta_{t-k}^i)\leq q_t\right)-\alpha^2\right\}}\Bigg/g_t(q_t).$ 
			\item {Since $\sigma_t^{K}$ is uniformly bounded for all $K$ and $t,$ we can see that outer-layer importance sampling improves the outer-layer convergence rate from  $O(\frac{1}{\sqrt{M}})$ to  $O(\frac{1}{\sqrt{KM}}),$ when compared with the Direct Monte Carlo method.}
		\end{itemize}
\end{remark}}

Now we show the convergence rate of the inner-layer CIS step in the following theorem.
\begin{theorem} \label{thm:in_rate}
	Under  Assumption \ref{assumption_Exp},  for any $\beta\in[0,1),$
	$$\lim_{M\rightarrow\infty}\lim_{N\rightarrow\infty} \sqrt{NM^\beta}\left(\hat{q}_{t}^{M,N,K}-\hat{q}_{t}^{M,K}\right)=0 \text{ in distribution. }$$
\end{theorem}

To compare the convergence rates with and without CIS, we further present a theorem to show the inner-layer convergence rate when CIS is not applied. Recall that the c.d.f and quantile estimators obtained without CIS are 
$$\breve{G}_{t}^{M,N,K}(h)=\frac{1}{KM}\sum_{k=0}^{K-1} \sum_{i=1}^M w_{t|t-k}^i \cI\left(h \geq \breve H^{N}(\theta_{t-k}^i)\right).$$
and
$$\breve{q}_{t}^{M,N,K}=\inf\left\{h: \breve{G}_t^{M,N,K}(h)\geq \alpha \right\},$$
respectively. Then
following the  similar line of proof, we can show that without CIS,  the convergence rate of $\breve{q}_{t}^{M,N,K}$ is of the order $O(1/\sqrt{N}).$ It is formally stated in the following theorem.
\begin{theorem}\label{thm:in_rate2}
	Suppose  Assumption \ref{assumption_Exp}.1 holds. Denote by $\theta_t$ the $\theta-$parameter corresponding to the quantile $q_t$ at time $t$, i.e., $H(\theta_t)=q_t.$ Then we have
	$$\lim_{M\rightarrow\infty}\lim_{N\rightarrow\infty} \sqrt{N}\left((\breve{q}_{t}^{M,N,K}-\hat{q}_{t}^{M,K}\right)=\sigma_tN(0,1) \text{ in distribution, }$$
	where  $\sigma_t^2=\Var_{\xi\sim p(\cdot|\theta_t)}\{h(\xi)\}.$
\end{theorem}

{ As $t$ goes to infinity, $\theta_t$ will finally converge to $\theta^c$ almost surely by the consistency of the posterior distributions, and thus $\sigma_t^2$ will converge to $\Var_{\xi\sim p(\cdot|\theta^c)}\{h(\xi)\},$ which implies the boundedness of $\sigma_t^2.$  The inner-layer convergence rate without CIS is  $O(1/\sqrt{N}).$} Note that since Theorem \ref{thm:in_rate} holds for  every $\beta\in[0,1),$ the convergence rate is approximately $O(1/\sqrt{NM}).$  {Thus, CIS improves the inner-layer convergence rate by a factor   $O(1/\sqrt{M}).$ }

{In general, compared with the Direct Monte Carlo method, the outer-layer importance sampling improves the outer-layer convergence rate from $O(1/\sqrt{M})$ to $O(1/\sqrt{KM})$, and the inner-layer importance sampling improves the inner-layer convergence rate from $O(1/\sqrt{N})$ to $O(1/\sqrt{NM}).$  Combining these two rates together, the overall convergence rate of TLIS is  $O(1/\sqrt{KM})+O(1/\sqrt{NM}).$ }

{Though the previous analysis on EFDs is sufficient to justify the applicability of our algorithm, studying the condition required for general distributions still has its own merit. Please refer to Appendix \ref{discuss_general} for further discussion on general distributions.}

\section{Numerical Experiments}\label{experiment}

In this section, we use the news vendor model as an example to demonstrate our proposed algorithms and compare with other algorithms including the Direct Monte Carlo method, the Simple Importance Sampling method, and the Green Simulation method. Consider a news vendor, who buys $q$ number of newspapers at the wholesale price $c$ each morning and then sells the newspapers throughout the day at a retail price $p$ higher than the wholesale price (i.e., $p>c$). At the end of the day, any unused papers can no longer be sold and are scrapped. Suppose the demand of the paper $D$ is a random variable following the exponential distribution with the true parameter $\theta^c.$

Given $q$, $p$, and $c$, the expected profit for the news vendor is
$$H(\theta^c)=\EE_{D\sim\exp(\theta^c)}[p\min\{q,D\}]-cq.$$
The problem here is that $\theta^c$ is unknown and has to be estimated using demand data, and thus the estimation error of $\theta^c$ would impact the estimation of the expected profit.

The demand data arrive sequentially in time. More specifically, starting from time $t=1$, there is one new data point $D_t$ arriving at each time stage $t$. All these data points are i.i.d. from the true input distribution, i.e., the exponential distribution with rate parameter $\theta^c$. We take a Bayesian approach to model the unknown true parameter $\theta^c$ and treat it as a random variable. We assume a non-informative Gamma prior, which is a conjugate prior of the exponential distribution, with shape parameter $0.001$ and scale parameter $1000$. Hence, at time stage $t,$ the posterior distribution $\pi_t$ given historic data $\{D_i\}_{i=1}^t$ is a Gamma distribution with shape parameter $t+0.001$ and scale parameter $1/(\sum_{i=1}^{t}D_i+0.001)$. Our goal here is to quantify the input uncertainty  by estimating the  $\alpha$-quantile of the posterior distribution of $H(\theta)$. Note that the true $\alpha-$quantile can be calculated analytically in this example. In fact, one can verify that $H(\theta)$ is strictly decreasing in $\theta$. If we denote by  $\theta_t^{1-\alpha}$ the $(1-\alpha)-$quantile of the Gamma posterior distribution  $\pi_t$, then the true $\alpha-$quantile of  $H$ at time $t$ is exactly $q_t:=H(\theta_t^{1-\alpha}).$ The true quantiles will be used to calculate the mean square errors (MSEs) of the quantile estimators. In our experiments stated below, we set $q=0.5$, $p=1.5$, $c=1$, $\theta^c=1$, $\alpha=0.05$ for lower quantile, and $\alpha=0.95$ for upper quantile.

\noindent {\bf Numerical Experiment 1.} We first consider Scenario I, where all simulations are done at time stage $t=0$ and no new simulation for all later stages $t\geq 1$. Algorithms that are applicable to this scenario include TLIS-1 and the Simple Importance Sampling method. Although the Direct Monte Carlo method is not applicable here due to its need for new simulations at each time stage, we still implement it as a benchmark. For fair comparison, all the algorithms use the same simulation budget. Specifically, the Direct Monte Carlo method runs $M*N$ simulation replications at every time stage, while TLIS-1 and the Simple Importance Sampling method run $T*M*N$ simulation replications only at $t=0.$  We set the outer sample size $M=30,$ inner sample size $N=10$, and time horizon $T=200$. For TLIS-1, we set the number of reused time stages $K=20.$  Note that if $t\leq20,$ TLIS-1 reuses the simulation outputs from all previous time stages.  We run $100$ macro replications, and report in Figure~\ref{exp1} the MSEs of both the upper and lower quantile estimates, which is computed at time $t$ according to $MSE_t = \frac{1}{100}\sum_{i=1}^{100}(\hat{q}^i_t-q_t)^2$, where $q_t$ is the true $\alpha$-quantile of the posterior distribution of $H(\theta)$, and $\hat{q}^i_t$ is the quantile estimate from the $i$-th macro replication.

\begin{figure}[t]
	\centering
	\includegraphics[scale=0.17]{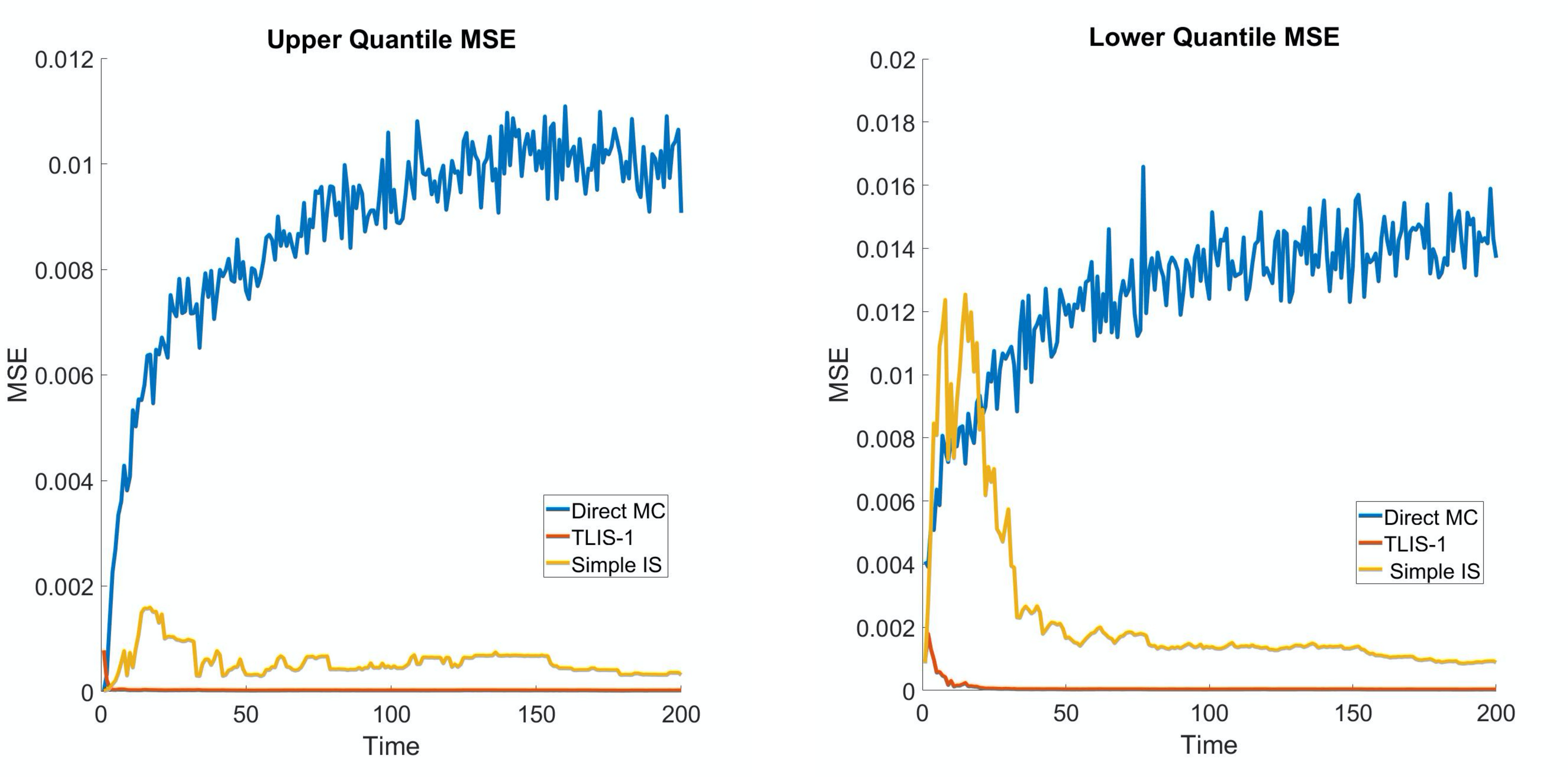}	
	\caption{Comparison between TLIS-1, Simple Importance Sampling (Simple IS) and the Direct Monte Carlo (Direct MC) method under the same simulation budget. $M=30,$ $N=10,$ and $T=200.$}
	\label{exp1}
\end{figure}

As Figure~\ref{exp1} shows,  TLIS-1 significantly outperforms the other two algorithms. Its estimators achieve small and stable MSEs that are close to 0 at all time stages. In contrast, the Direct Monte Carlo method and Simple Importance Sampling method have much larger MSEs. The fact that Simple Importance Sampling performs worse than TLIS-1 justifies that drawing new out-layer  $\theta-$samples in every time stage helps achieve a more accurate estimate. This can be done under a limited simulation budget, because our inner-layer CIS makes it possible to estimate the system performance under new  $\theta-$parameter without doing any new simulations. Moreover, Direct Monte Carlo performs the worst under the limited budget due to the large inner-layer estimation error. Thus, TLIS-1 is the more preferred method  when the total simulation budget is limited and fast decision-making is required.
\begin{figure}[t]
	\centering
		\includegraphics[scale=0.19]{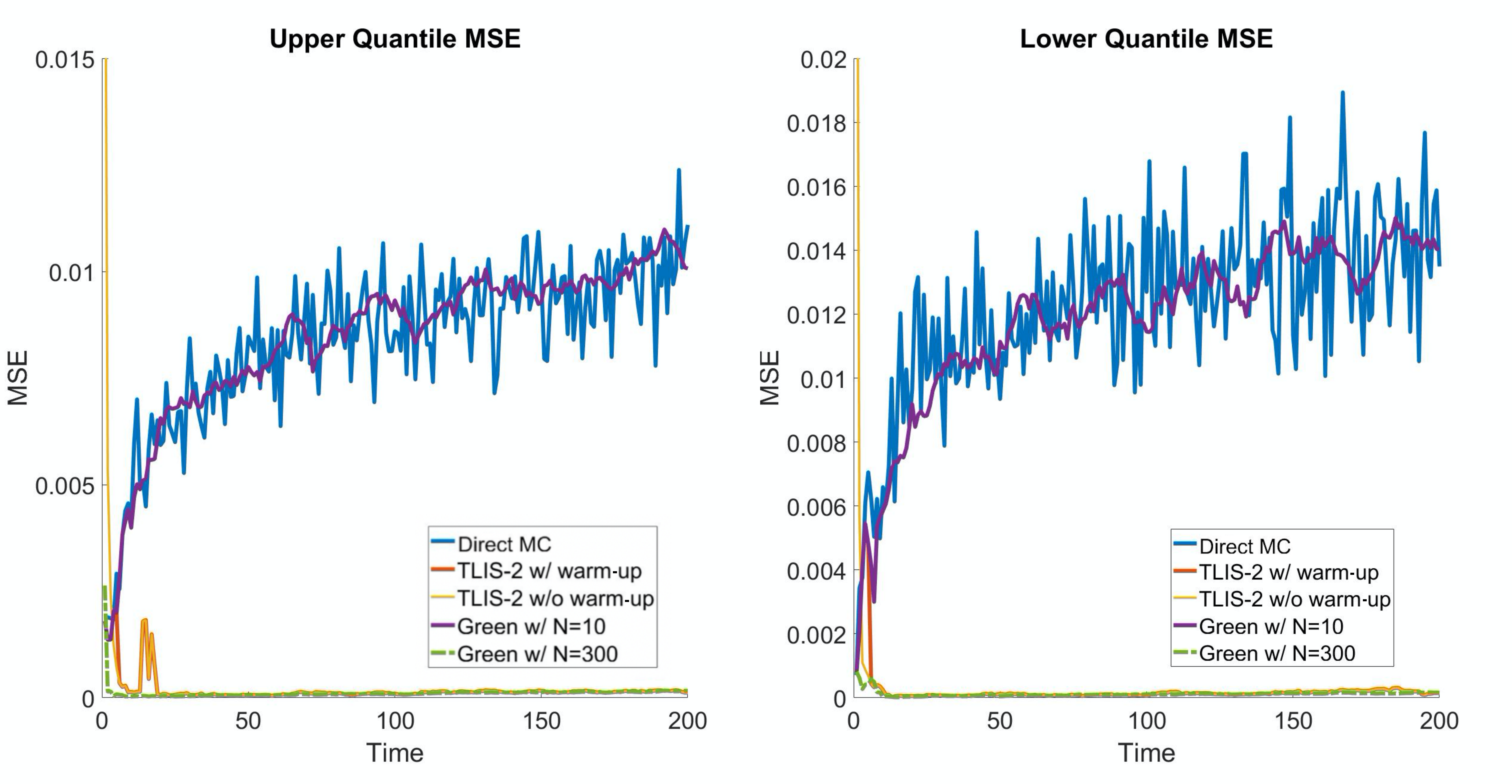}		
	\caption{Comparison between TLIS-2 with and without warm up ($N=10$), Green Simulation (Green) with $N=10$ and $N=300$  and the Direct Monte Carlo (Direct MC) method ($N=10$). }
	\label{exp2}
\end{figure}

\noindent {\bf Numerical Experiment 2.} We next consider Scenario II, where we can afford a small number of new simulations at each time stage. Algorithms that are applicable to this scenario include TLIS-2, the Green simulation method, and the Direct Monte Carlo method. For fair comparison, all the algorithms are run using the same simulation budget. Specifically, these three algorithms run $M\times N$ simulation experiments at any time stage.  We set the outer sample size $M=30,$ inner sample size $N=10$, and time horizon $T=200$. To see clearly the improvement brought by CIS, we also run Green Simulation with $N=300$, which is equal to the effective inner sample size $M*N = 30*10$ when using CIS. We remark that when we have only a few data points  at the first few time stages, the posterior distribution usually has a large variance. In this case, when the sample number is relatively small (here is $M=30$), any two outer-layer $\theta-$samples can be very different, and hence, when using CIS, the large variance among the importance weights could lead to a large estimation error. {Therefore, when applying TLIS-2, we recommend  warming up the algorithm at the initial time stages where we do not use CIS, although using it will not impact the later stages. This echos the analysis after Theorem~\ref{thm:exp_in_bound}. In this experiment, we run TLIS-2 with and without warm-up. TLIS-2 with warm-up starts applying CIS after the 5th time stage.} We run 100 macro replications and report the MSEs of both the upper and lower quantile estimates over time in Figure \ref{exp2}.

We have the following observations.
\begin{itemize}
	\item
	As shown in Figure~\ref{exp2}, the estimates obtained by TLIS-2 are much more precise than the others, including the Green Simulation method, which clearly shows the benefit of CIS. 
	\item {TLIS-2 without warm-up performs the worst among all the algorithms at initial time stages. However, after several time stages, TLIS-2 with and without warm-up have similar good performance. }
	\item { Under a small simulation budget, Green Simulation achieves similar MSE as the Direct Monte Carlo method.  In fact, without CIS, the inner-layer simulation uncertainty dominates the MSE and ruins the performance of these two algorithms when $N$ is very small. Thus, our CIS technique is crucial under limited budget. Moreover, the trajectory of Green Simulation is more stable than the Direct Monte Carlo. This justifies that reusing $\theta-$parameters from previous time stages can greatly reduce the variance in the quantile estimator.}
	
	\item TLIS-2 with $N=10$ performs similar to Green Simulation with $N=300$, which shows the effectiveness of CIS that prompts the effective inner sample size of TLIS-2 by $M=30$ times.  This verifies our theoretical analysis in Theorem \ref{thm:in_rate}.
\end{itemize}

\noindent {\bf Numerical Experiment 3.}
In this experiment, we  study the choice of $(M,N)$ (with a fixed $K$) under a fixed simulation budget $B$ at each time stage. Recall that in Scenario II we carry out $N\times M$ new simulation replications at each time stage, and hence, $N\times M=B$. According to Theorem \ref{thm:out_rate} and \ref{thm:in_rate}, the estimation error is approximately of order $O(1/\sqrt{KM})+O(1/\sqrt{NM})=O(1/\sqrt{KM})+O(1/\sqrt{B}).$ Given the fixed budget $B$ and fixed time horizon $T$, the order does not explicitly depend on $N,$ which means we usually need to choose a large $M$ and small $N$. This is mainly because CIS utilizes simulation outputs from other input parameters, resulting the number of simulation outputs used in each estimator equivalent to $NM=B.$  This analysis shows that when the budget is fixed, we should favor a large $M$ and a small $N$. However, $N=1$  (the smallest possible value) usually does not work well in practice, since our convergence results are in the  asymptotic sense (i.e., $N$ and $M$ should be sufficiently large).
\begin{table}[]
	\centering
	\begin{tabular}{|c|c|c|c|c|c|c|c|c|}
		\cline{1-9}
		& \multicolumn{4}{c|}{Lower MSE ($\times 10^{-3}$)}                                        & \multicolumn{4}{c|}{Upper MSE ($\times 10^{-3}$)}                                                                      \\ \hline
		\multicolumn{1}{|c|}{}    & $t=50 $           & $t=100 $          & $t=150$           & $t=200$           &$ t=50$            & $t=100$                          & $t=150$                          & $t=200$           \\ \hline
		\multicolumn{1}{|c|}{$M=10$} & 0.3348          &  0.2094          & 0.1418          & 0.2173          & 0.2834          & 0.2310 & 0.1635 & 0.3320          \\ \hline
		\multicolumn{1}{|c|}{$M=30$} &  0.1203         & 0.1779          & 0.1558          & 0.1321        & 0.1603          & 0.1001                       & 0.1529 & 0.2451         \\ \hline
		\multicolumn{1}{|c|}{$M=50$} & \textbf{0.0806} & \textbf{ 0.1153} & \textbf{0.0706} & \textbf{0.1296
		} & \textbf{0.0645} & \textbf{0.0767}                & \textbf{0.1076}                & \textbf{0.2016} \\ \hline
	\end{tabular}
	\caption{Mean Square Errors of quantiles estimatiors  by TLIS-2 with $(M,N)=(10,30),(30,10),(50,6).$}
	\label{tab_exp3}
\end{table}

We next empirically compare different choices of $(M,N)$ under the same simulation budget. We use the same setting in Experiment 2 except that we consider three choices of $$(M,N)\in\{(50,6), (30,10),(10,30)\}.$$ We run the experiment for 100 times and report the MSEs of upper and lower quantile estimates at time $t=100,150,200$ in Table \ref{tab_exp3}. Here, we observe that choosing $M=50$ achieves smaller MSEs for both lower and upper quantile estimates at all 4 time stages. This is consistent with our argument above that using a large $M$ and small $N$ may achieve the best performance.  $M=30$ works  better than $M=10$ most of the time except for the lower quantile estimate at $t=150.$  Since our theoretical results are in the asymptotic sense, this slight inconsistency when $M$ and $N$ are small is reasonable and acceptable. 

\noindent {\bf Numerical Experiment 4.} In this experiment, we empirically study how to choose $K$, the number of time stages reused, to obtain the best accuracy and efficiency.  From Theorem \ref{thm:out_rate}, we see clearly $K$ only affects the convergence rate of the outer layer, and the larger $K$ is, the better the estimator is.  However, this only holds when the variance of the importance weights is bounded. Unfortunately, the following result shows that when $K=t$ the variance will explode as $t$ goes to infinity.
\begin{align}
\EE\left(\frac{\pi_{t}(\theta)}{\pi_{0}(\theta) }\right)^2&=\int_\Theta{ f(\chi_{t},t)^2  \pi_0(\theta)\exp \left (2\theta^\top \boldsymbol\chi_{t}- 2t A(\theta)\right )}d\theta=O_{a.s.}\left((\sqrt{\frac{t}{2}})^s\right).\label{exlpode}
\end{align}
That means when we choose $K=t$, to obtain convergence we need an extremely large $M$  to control the variance, especially when time $t$ and dimension $s$ are large. Moreover, \eqref{exlpode} also suggests that $K$ is not allowed to increase in the same speed as $t$. However, one can still  choose a fixed large constant  $K$, since the posterior distribution does not change much from time stage $t-K$ to $t$ when $t$ is large.

In this experiment, we test $K=10,~50,~100$, and $200$ to verify our discussion above. To focus only on the outer-layer estimation, we set a large inner-layer sample size $N=1000$ to make the inner-layer simulation error negligible. We run the experiment for 100 times and report the MSEs of upper and lower quantile estimates at time $t=100,150,200$ in Table \ref{tab_exp41}. As shown in Table~\ref{tab_exp41}, reusing outer-layer samples from more previous time stages does not necessarily lead to a better estimate.  Specifically, we find that reusing the latest $K=100$ time stages  performs better than $K=200$  at $t=200.$  Though we use more samples when $K=200$,   $\pi_{t-K}$ is not a good proposal distribution for $\pi_t$ when $K$ is large,  and thus reusing samples from $\pi_{t-K}$ may not bring any benefit. Moreover, the average  running time of the last 100 iterations is also presented in Table \ref{tab_exp41}. Note that when $K=200,$ the running time (8.9s)  is about $1.5$ times of that of $K=100$ (5.2s). However, $K=200$ does not yield better performance than $K=100$.  Thus, when considering both running time and estimation accuracy, a large $K$ is not necessarily a good choice.
\begin{table}[]
	\centering
	\begin{tabular}{|c|c|c|c|c|c|c|c|}
		\hline
		& \multicolumn{3}{c|}{Lower MSE  ($\times 10^{-3}$)}                      & \multicolumn{3}{c|}{Upper MSE  ($\times 10^{-3}$)}                                                    &              \\ \hline
		& $t=100$           &$ t=150  $         &$ t=200 $          &$ t=100  $         &$ t=150    $                      &$ t=200 $                         & Running Time \\ \hline
		$K=10$  & 0.2967         & 0.0442         & 0.0474         & 0.135         &0.7777 &0.0132 & 1.5s         \\ \hline
		$K=50$  & 0.0700        & 0.0411          & 0.0571          & 0.0126          & 0.0187                        & 0.0022 & 2.3s         \\ \hline
		$K=100$ & \textbf{0.0337} & \textbf{0.0341} & \textbf{0.0401} &\textbf{0.0090} & 0.01807                         & \textbf{0.0018}                & 5.2s         \\ \hline
		$K=200$ & \textbf{0.0337} & 0.0665          & 0.0640          & \textbf{0.0090} & \textbf{0.0155}                & 0.0054                         & 8.9s         \\ \hline
	\end{tabular}
	\caption{ MSEs of  quantiles estimators  and running time of the last 100 iterations of Two-layer Importance Sampling with $K=10/50/100/200.$}
	\label{tab_exp41}
\end{table}

\section{Conclusions}
 {In this paper, we propose a Two-layer Importance Sampling (TLIS) method to quantify input uncertainty in an online manner with streaming data. This method uses importance sampling to reuse simulation outputs from previous time stages in the outer-layer and the simulation outputs under other input parameter samples in the inner-layer. Meanwhile, to  meet the requirement of different applications, we design two algorithm versions under the TLIS framework to suit for two special scenarios where fast or moderate speed of decision making is required. We prove the asymptotic convergence and convergence rate results, which show the speed up of the  TLIS framework. Numerical examples are conducted to justify our theoretical analysis and demonstrate the superior performance of TLIS over other existing methods.}

\section*{Acknowledgment}
{The authors  are grateful to the support by the National Science Foundation under Grant CAREER CMMI-1453934, and Air Force Office of Scientific Research under Grant FA9550-19-1-0283.
	This work is based on our preliminary work {\it Online Quantification of Input Uncertainty in Parametric Models} which has been published in the Proceedings of the 2018 Winter Simulation Conference. }
\bibliographystyle{ims}
\bibliography{draft}

\begin{thebibliography}{25}
\expandafter\ifx\csname natexlab\endcsname\relax\def\natexlab#1{#1}\fi
\expandafter\ifx\csname url\endcsname\relax
  \def\url#1{\texttt{#1}}\fi
\expandafter\ifx\csname urlprefix\endcsname\relax\def\urlprefix{URL }\fi

\bibitem[{Barton(2012)}]{barton2012tutorial}
\textsc{Barton, R.~R.} (2012).
\newblock Tutorial: Input uncertainty in output analysis.
\newblock In \textit{Proceedings of the 2012 Winter Simulation Conference}.
  IEEE, Berling, Germany.

\bibitem[{Barton et~al.(2013)Barton, Nelson and Xie}]{barton2013quantifying}
\textsc{Barton, R.~R.}, \textsc{Nelson, B.~L.} and \textsc{Xie, W.} (2013).
\newblock Quantifying input uncertainty via simulation confidence intervals.
\newblock \textit{INFORMS Journal on Computing} \textbf{26} 74--87.

\bibitem[{Barton and Schruben(1993)}]{barton1993uniform}
\textsc{Barton, R.~R.} and \textsc{Schruben, L.~W.} (1993).
\newblock Uniform and bootstrap resampling of empirical distributions.
\newblock In \textit{Proceedings of the 1993 Winter Simulation Conference}.
  IEEE, Los Angeles, California.

\bibitem[{Barton and Schruben(2001)}]{barton2001resampling}
\textsc{Barton, R.~R.} and \textsc{Schruben, L.~W.} (2001).
\newblock Resampling methods for input modeling.
\newblock In \textit{Proceedings of the 2001 Winter Simulation Conference}.
  IEEE, Arlington, Virginia.

\bibitem[{Cheng and Holloand(1997)}]{cheng1997sensitivity}
\textsc{Cheng, R.~C.} and \textsc{Holloand, W.} (1997).
\newblock Sensitivity of computer simulation experiments to errors in input
  data.
\newblock \textit{Journal of Statistical Computation and Simulation}
  \textbf{57} 219--241.

\bibitem[{Chick(2001)}]{chick2001input}
\textsc{Chick, S.~E.} (2001).
\newblock Input distribution selection for simulation experiments: accounting
  for input uncertainty.
\newblock \textit{Operations Research} \textbf{49} 744--758.

\bibitem[{Durrett(2019)}]{durrett2019probability}
\textsc{Durrett, R.} (2019).
\newblock \textit{Probability: theory and examples}, vol.~49.
\newblock Cambridge University Press.

\bibitem[{Egloff and Leippold(2010)}]{egloff2010quantile}
\textsc{Egloff, D.} and \textsc{Leippold, M.} (2010).
\newblock Quantile estimation with adaptive importance sampling.
\newblock \textit{The Annals of Statistics} \textbf{38} 1244--1278.

\bibitem[{Feng and Song(2019)}]{feng2019WSCuq}
\textsc{Feng, M.} and \textsc{Song, E.} (2019).
\newblock Efficient input uncertainty quantification via green simulation using
  sample path likelihood ratios.
\newblock In \textit{Proceedings of the 2019 Winter Simulation Conference}.
  IEEE, National Harbor, Maryland.

\bibitem[{Feng and Staum(2015)}]{feng:2015WSCgreensim}
\textsc{Feng, M.} and \textsc{Staum, J.} (2015).
\newblock Green simulation designs for repeated experiments.
\newblock In \textit{Proceedings of the 2015 Winter Simulation Conference}.
  IEEE, Huntington Beach, California.

\bibitem[{Feng and Staum(2017)}]{Feng2017GreenSR}
\textsc{Feng, M.} and \textsc{Staum, J.} (2017).
\newblock Green simulation: reusing the output of repeated experiments.
\newblock \textit{ACM Transactions on Modeling and Computer Simulation
  (TOMACS)} \textbf{27} 1--28.

\bibitem[{Glynn(1996)}]{glynn1996importance}
\textsc{Glynn, P.~W.} (1996).
\newblock Importance sampling for monte carlo estimation of quantiles.
\newblock In \textit{Mathematical Methods in Stochastic Simulation and
  Experimental Design: Proceedings of the 2nd St. Petersburg Workshop on
  Simulation}. Publishing House of St. Petersburg University.

\bibitem[{Johnson et~al.(1967)}]{johnson1967asymptotic}
\textsc{Johnson, R.~A.} \textsc{et~al.} (1967).
\newblock An asymptotic expansion for posterior distributions.
\newblock \textit{The Annals of Mathematical Statistics} \textbf{38}
  1899--1906.

\bibitem[{Lam(2016)}]{lam:2016sensitivity}
\textsc{Lam, H.} (2016).
\newblock Robust sensitivity analysis for stochastic systems.
\newblock \textit{Mathematics of Operations Research} \textbf{41} 1248--1275.

\bibitem[{Lam and Qian(2016)}]{lam2016empirical}
\textsc{Lam, H.} and \textsc{Qian, H.} (2016).
\newblock The empirical likelihood approach to simulation input uncertainty.
\newblock In \textit{Proceedings of the 2016 Winter Simulation Conference
  (WSC)}. IEEE, Arlington, Virginia.

\bibitem[{Lam and Qian(2018)}]{lam2018subsampling}
\textsc{Lam, H.} and \textsc{Qian, H.} (2018).
\newblock Subsampling variance for input uncertainty quantification.
\newblock In \textit{2018 Winter Simulation Conference (WSC)}. IEEE,
  Gothenburg, Sweden.

\bibitem[{Lin et~al.(2015)Lin, Song and Nelson}]{lin2015single}
\textsc{Lin, Y.}, \textsc{Song, E.} and \textsc{Nelson, B.~L.} (2015).
\newblock Single-experiment input uncertainty.
\newblock \textit{Journal of Simulation} \textbf{9} 249--259.

\bibitem[{Schwartz(1965)}]{schwartz1965bayes}
\textsc{Schwartz, L.} (1965).
\newblock On bayes procedures.
\newblock \textit{Zeitschrift f{\"u}r Wahrscheinlichkeitstheorie und verwandte
  Gebiete} \textbf{4} 10--26.

\bibitem[{Song and Nelson(2017)}]{song2017input}
\textsc{Song, E.} and \textsc{Nelson, B.~L.} (2017).
\newblock Input model risk.
\newblock In \textit{Advances in Modeling and Simulation}, chap.~5. Springer,
  63--80.

\bibitem[{Van~der Vaart(2000)}]{van2000asymptotic}
\textsc{Van~der Vaart, A.~W.} (2000).
\newblock \textit{Asymptotic statistics}, vol.~3.
\newblock Cambridge University Press.

\bibitem[{Wu et~al.(2018)Wu, Zhu and Zhou}]{wu2018bayesian}
\textsc{Wu, D.}, \textsc{Zhu, H.} and \textsc{Zhou, E.} (2018).
\newblock A bayesian risk approach to data-driven stochastic optimization:
  Formulations and asymptotics.
\newblock \textit{SIAM Journal on Optimization} \textbf{28} 1588--1612.

\bibitem[{Xie et~al.(2014)Xie, Nelson and Barton}]{xie2014bayesian}
\textsc{Xie, W.}, \textsc{Nelson, B.~L.} and \textsc{Barton, R.~R.} (2014).
\newblock A bayesian framework for quantifying uncertainty in stochastic
  simulation.
\newblock \textit{Operations Research} \textbf{62} 1439--1452.

\bibitem[{Zhu et~al.(2019)Zhu, Liu and Zhou}]{ZhouZhu:2015risk}
\textsc{Zhu, H.}, \textsc{Liu, T.} and \textsc{Zhou, E.} (2019).
\newblock Risk quantification in stochastic simulation under input uncertainty.
\newblock \textit{ACM Transactions on Modeling and Computer Simulation}
  Accepted.

\bibitem[{Zouaoui and Wilson(2003)}]{zouaoui2003accounting}
\textsc{Zouaoui, F.} and \textsc{Wilson, J.~R.} (2003).
\newblock Accounting for parameter uncertainty in simulation input modeling.
\newblock \textit{IIE Transactions} \textbf{35} 781--792.

\bibitem[{Zouaoui and Wilson(2004)}]{zouaoui2004accounting}
\textsc{Zouaoui, F.} and \textsc{Wilson, J.~R.} (2004).
\newblock Accounting for input-model and input-parameter uncertainties in
  simulation.
\newblock \textit{IIE Transactions} \textbf{36} 1135--1151.

\end{thebibliography}

\newpage
\appendix



\section{Proof of Important Properties of Exponential Family}\label{exp_proof}
\subsection{Proof of Theorem \ref{thm:exp_out_bounded}}\label{proof:exp_out_bounded}
\begin{proof}
	We first compute the variance of the weight using $\pi_{t_2}(\theta) $ as the proposal distribution and $\pi_{t_1}(\theta)$ as the target distribution, where $t_1>t_2>0.$
	\begin{align*}
	\EE\left(\frac{\pi_{t_1}(\theta)}{\pi_{t_2}(\theta) }\right)^2&=\int_\Theta \frac{\pi_{t_1}(\theta)^2}{\pi_{t_2}(\theta) }d \theta\\
	&=\int_\Theta\frac{ f(\boldsymbol\chi_{t_1},t_1)^2  \pi_0(\theta)^2\exp \left (2\theta^\top \boldsymbol\chi_{t_1}- 2t_1 A(\theta)\right )}{f(\boldsymbol\chi_{t_2},t_2) \pi_0(\theta)\exp \left (\theta^\top \boldsymbol\chi_{t_2} - t_2 A(\theta)\right )} d\theta\\
	&=\int_\Theta\frac{ f(\boldsymbol\chi_{t_1},t_1)^2 }{f(\boldsymbol\chi_{t_2},t_2) } \pi_0(\theta)\exp \left (\theta^\top (2\boldsymbol\chi_{t_1}-\chi_{t_2})-(2t_1-t_2) A(\theta) \right ) d\theta\\
	&=\int_\Theta\frac{ f(\boldsymbol\chi_{t_1},t_1)^2 }{f(\boldsymbol\chi_{t_2},t_2) f(2\boldsymbol\chi_{t_1}-\chi_{t_2},2t_1-t_2)}f(2\boldsymbol\chi_{t_1}-\chi_{t_2},2t_1-t_2)\exp \left (\theta^\top (2\boldsymbol\chi_{t_1}-\chi_{t_2} ) - (2t_1-t_2)A(\theta) \right ) d\theta\\
	&=\frac{ f(\boldsymbol\chi_{t_1},t_1)^2 }{f(\boldsymbol\chi_{t_2},t_2) f(2\boldsymbol\chi_{t_1}-\chi_{t_2},2t_1-t_2)}\\
	&=\frac{\int_\Theta \pi_0(\theta) \exp \left (\theta^\top \boldsymbol\chi_{t_2} - t_2 A(\theta) \right )d\theta \int_\Theta  \pi_0(\theta)\exp \left (\theta^\top (2\boldsymbol\chi_{t_1}-\chi_{t_2}) - (2t_1-t_2)A(\theta) \right )d\theta}{(\int_\Theta  \pi_0(\theta)\exp \left (\theta^\top \boldsymbol\chi_{t_1} - t_1 A(\theta) \right )d\theta)^2}.
	\end{align*}
	Similarly, we have
	\begin{align*}
	\EE\left(\frac{\pi_{t_1}(\theta)}{\pi_{t_2}(\theta) }\right)^3&=\int_\Theta \frac{\pi_{t_1}(\theta)^3}{\pi_{t_2}(\theta)^2 }d \theta\\
	&=\frac{ f(\boldsymbol\chi_{t_1},t_1)^3 }{f(\boldsymbol\chi_{t_2},t_2)^2 f(3\boldsymbol\chi_{t_1}-2\chi_{t_2},3t_1-2t_2)}<\infty.\end{align*}
	Now we calculate the second moment. Let's first approximate the following integral when $n$ is  large. 
	$$\int_\Theta  \pi_0(\theta)\exp \left(\left (\theta^\top( \boldsymbol\chi_{n}) - nA(\theta) \right )\right)d\theta. $$
	From the properties of $A(\theta),$ we know$$\nabla^2A(\theta)=\Var\{{T}(\xi)\}, \nabla A(\theta)=\EE\{{T}(\xi)\}.$$
	Under the assumption that  $\nabla^2A(\theta)$ is positive definite, by the concentration property of the posterior distribution with large samples,  we  can calculate the integral  $\int_\Theta \pi_0(\theta) \exp \left (\theta^\top \sum_{i=1}^n {T}(\xi_i) - n A(\theta) \right )d\theta$ in a small neighborhood around the maximum likelihood estimator  (MLE) $\hat\theta_n.$ 
	Note that our uniform prior  guarantees that $\pi_0$ is bounded, i.e.,$ \pi_0(\theta)\leq C$ for some constant $C\geq0.$
	Note that for EFDs, the MLE is strong consistent, i.e.,
	$$\hat{\theta}_n\rightarrow \theta^c,~a.s.(\PP^\NN_{\theta^c})~ \text{as }n\rightarrow\infty,$$
	where $\hat{\theta}_n$ is the MLE of $\theta$ given data $\frac{1}{n} \sum_{i=1}^n {T}(\xi_i).$
	
	For simplicity, we define $\boldsymbol\eta=(\nabla^2 A(\hat{\theta}(r)))^{\frac{1}{2}}(\theta-\hat{\theta}(r)),$ where $\hat{\theta}(r)$ is the MLE given  $r.$ In our case $ r=\frac{1}{n} \sum_{i=1}^n {T}(\xi_i).$ Let 
	$$g(\boldsymbol\eta,r)=\exp(A(\hat{\theta}(r))-A(\hat{\theta}(r)+(\nabla^2 A(\hat{\theta}(r)))^{-\frac{1}{2}}\boldsymbol\eta)+ r^\top(\nabla^2 A(\hat{\theta}(r)))^{-\frac{1}{2}}\boldsymbol\eta),$$
	and $$h(\boldsymbol\eta,r)=\ln g(\boldsymbol\eta,r).$$
It is obvious that for a fixed $r$, $g(0,r)=1,$ $\nabla g(0,r)=0,$ $\nabla^2 g(0,r)=-I,$ since $\nabla A((\hat{\theta}(r))=r$ (MLE's property). Then we have the following  lemma.
	\begin{lemma}\label{lemma_33}
	Assume $\nabla^2A(\theta)$ is positive definite for all $\theta\in\Theta$, then we have
		$$\int_{E_{\hat{\theta}_n}} \pi_0(\hat{\theta}_n+(\nabla^2 A(\hat{\theta}_n))^{-\frac{1}{2}}\boldsymbol\eta)g^n(\boldsymbol\eta,\frac{1}{n} \sum_{i=1}^n {T}(\xi_i)) d\boldsymbol\eta=O_{\text{a.s.}}((\frac{1}{\sqrt{n}})^s),$$
		where $E_{\hat{\theta}_n}=(\nabla^2 A(\hat{\theta}_n))^{\frac{1}{2}}(\Theta-\hat{\theta}_n))$ is the parameter space after transformation.  
	\end{lemma}
	This lemma is a direct extension of  Lemma 3.3  in \cite{johnson1967asymptotic} to the multi-parameter EFDs.
	\begin{proof}
		Denote $r_0=\nabla A({\theta^c}).$ Thus, we have  $ g(0,r_0)=1.$ First, we consider the function $$\phi(\theta)=(\theta)\exp{\theta^\top r-A(\theta)}.$$ Since $A(\theta)$ is continuous, $\phi(\theta)$ is continuous. The following lemma shows that under certain conditions, $\phi(\theta)$ has a unique maximum. 
		\begin{lemma} \label{lem_maximum}
			There exists a $d_1>0$ such that for fixed $\norm{r-r_0}_2\leq d_1,$ the function $\phi(\theta)=\exp{\theta^\top r-A(\theta)}$ has a unique maximum at $\hat{\theta}(r).$ And for any vector $v\in \RR^s,$ we have for $0<t_1<t_2,$
			$$\exp{ (\hat{\theta}(r)+t_1v)^\top r-A(\hat{\theta}(r)+t_1v))}>\exp{ (\hat{\theta}(r)+t_2v)^\top r-A(\hat{\theta}(r)+t_2v))}.$$
		\end{lemma}
		\begin{proof}
			Since $A(\theta)$ is a convex function of $\theta$, $\hat{\theta}(r)\rightarrow \theta^c$ as $r\rightarrow r_0$ and $A(\theta^c)<\infty,$  there exists a $d_1>0$ such that $A(\hat{\theta}(r))<\infty$ holds for all  $r$ satisfying $\norm{r-r_0}_2\leq d_1.$
			Then the unique maximum comes from the fact   $\nabla A(\hat{\theta}(r))=r,$ and $\nabla^2 A({\theta})$ is positive definite. The positive definiteness also implies the strictly concavity of the function. Thus, we have 
			\begin{align}\label{lem_monotone}
			\exp{ (\hat{\theta}(r)+t_1v)^\top r-A(\hat{\theta}(r)+t_1v))}>\exp{ (\hat{\theta}(r)+t_2v)^\top r-A(\hat{\theta}(r)+t_2v))} 
			\end{align}
			when $A(\hat{\theta}(r)+t_1v)$ and  $A(\hat{\theta}(r)+t_2v)$ are finite.
		\end{proof}
			\noindent For a fixed $r\in \{r:\norm{r-r_0}_2\leq d_1\}$ and $\norm{\boldsymbol\eta}_2\leq 1,$  Taylor  expansion for $h(\boldsymbol\eta,r)$ around $h(0,r)$ yields
		$$h(\boldsymbol\eta,r)=-\frac{1}{2}\boldsymbol\eta^\top \boldsymbol\eta+R_2(\boldsymbol\eta),$$
		where the remaining term $R_2(\boldsymbol\eta)$ satisfies the following inequality,
		$$|R_2(\boldsymbol\eta)|\leq\frac{C}{6}\norm{\boldsymbol\eta}_2^3,$$
		where $C$ is some constant, since$ \nabla^3h(\boldsymbol\eta,r)$ is uniformly continuous in $\left\{(\boldsymbol\eta,r)|\norm{\boldsymbol\eta}_2\leq1, \norm{r-r_0}_2\leq d_1\right\}.$ 
		When $\norm{\boldsymbol\eta}_2\leq \delta \ll\frac{1}{C},$ we have $|R_2(\boldsymbol\eta)|\ll\boldsymbol\eta^\top \boldsymbol\eta.$
		Thus, we have $$h(\boldsymbol\eta,r)\approx-\frac{1}{2}\boldsymbol\eta^\top \boldsymbol\eta,~\text{if}~ \norm{\boldsymbol\eta}_2\leq \delta.$$
		We finish the proof of  Lemma \ref{lem_maximum}. 
		
			\noindent Next we bound the value of $g(\boldsymbol\eta,r)$ when $\norm{\boldsymbol\eta}_2\geq \delta$ by the following lemma.		
		\begin{lemma}\label{lem_15}
			There exists an $0<\epsilon<1 ,$ such that 
			$g(\boldsymbol\eta,r)\leq \epsilon$ for all $\boldsymbol\eta$ satisfying $\norm{\boldsymbol\eta}_2\geq \delta.$
		\end{lemma}
		\begin{proof}
			By Lemma \ref{lem_maximum},  when $||\boldsymbol\eta||_2=\sqrt{\sum_i|\boldsymbol\eta^{(i)}|^2}\geq \delta/2$ and $\norm{r-r_0}_2\leq d_1$ we have $ g(\boldsymbol\eta,r)<1.$   Thus, take $\delta$ small enough and for each fixed $r$, we have by \eqref{lem_monotone}
			$$g(\boldsymbol\eta,r)\leq  \sup_{\boldsymbol\eta\in \SSS_0(\delta/2)}\{g(\boldsymbol\eta,r)\},~\forall \norm{\boldsymbol\eta}_2\geq \delta.$$
			For notational simplicity, we denote $\BB_0(r)=\left\{x\in \RR^s\big|\norm{x}_2\leq r\right\}$ as the ball with radius $r,$ and $ \SSS_0(r)=\left\{x\in \RR^s\big|\norm{x}_2=r\right\}$ as its sphere.  Since $ \SSS_0(\delta/2)$ is compact, there exists $\boldsymbol\eta'\in\SSS_0(\delta/2),$ such that  $g(\boldsymbol\eta',r)=\sup_{\boldsymbol\eta\in \SSS_0(\delta/2)}\{g(\boldsymbol\eta,r)\}<g(0,r)=1.$
			By the same argument, we have $\epsilon=\sup_{\norm{r-r_0}_2\leq d_0} \sup_{\boldsymbol\eta\in \SSS_0(\delta/2)}\{g(\boldsymbol\eta,r)\}<1.$
			Thus $$g(\boldsymbol\eta,r)\leq  \epsilon,~\forall \norm{\boldsymbol\eta}_2\geq \delta,~\norm{r-r_0}_2\leq d_1.$$
			We finish the proof of  Lemma \ref{lem_15}. 

		\end{proof}
		Note that by the law of large number,  when $n$ is large enough,  $$\norm{\frac{1}{n} \sum_{i=1}^n {T}(\xi_i)-r_0}_2\leq d_1,~~~\text{a.s.}$$
		By Lemma \ref{lem_15}, we have
		\begin{align}
		\int_{E_{\hat{\theta}_n}\backslash \BB_0(\delta)} \pi_0(\hat{\theta}_n+(\nabla^2 A(\hat{\theta}_n))^{-\frac{1}{2}}\boldsymbol\eta)g^n(\boldsymbol\eta,\frac{1}{n} \sum_{i=1}^n {T}(\xi_i)) d\boldsymbol\eta\leq \epsilon^n.
		\end{align}
		Furthermore, in the ball $\BB_0(\delta),$ $g(\boldsymbol\eta,\frac{1}{n} \sum_{i=1}^n {T}(\xi_i))$ can be approximated by $\exp(-\frac{1}{2}\boldsymbol\eta^\top\boldsymbol\eta).$ Thus,
		\begin{align*}
		& \int_{ \BB_0(\delta)} \pi_0(\hat{\theta}_n+(\nabla^2 A(\hat{\theta}_n))^{-\frac{1}{2}}\boldsymbol\eta)g^n(\boldsymbol\eta,\frac{1}{n} \sum_{i=1}^n {T}(\xi_i)) d\boldsymbol\eta\\
		&\approx \int_{ \BB_0(\delta)}\pi_0(\hat{\theta}_n+(\nabla^2 A(\hat{\theta}_n))^{-\frac{1}{2}}\boldsymbol\eta)\exp(-\frac{n}{2}\boldsymbol\eta^\top\boldsymbol\eta) d\boldsymbol\eta\\
		&= \pi_0(\boldsymbol \theta^c) \int_{ \BB_0(\delta)}\exp(-\frac{n}{2}\boldsymbol\eta^\top\boldsymbol\eta) d\boldsymbol\eta\\
		&=(\frac{1}{\sqrt{n}})^s  \pi_0(\boldsymbol \theta^c)\int_{ \BB_0(\sqrt{n}\delta)}\exp(-\frac{1}{2}\boldsymbol\eta^\top\boldsymbol\eta) d\boldsymbol\eta\\
		&\approx (\frac{1}{\sqrt{n}})^s\pi_0(\boldsymbol \theta^c)\int_{\RR^s}\exp(-\frac{1}{2}\boldsymbol\eta^\top\boldsymbol\eta) d\boldsymbol\eta.~\quad\quad\text{(when $n$ is large.)}
		\end{align*}
		Since $\epsilon^n=o((\frac{1}{\sqrt{n}})^s),$ we have $$\int_{E_{\hat{\theta}_n}} \pi_0(\hat{\theta}_n+(\nabla^2 A(\hat{\theta}_n))^{-\frac{1}{2}}\boldsymbol\eta)g^n(\boldsymbol\eta,\frac{1}{n} \sum_{i=1}^n {T}(\xi_i)) d\boldsymbol\eta= O_{a.s}((\frac{1}{\sqrt{n}})^s).$$
		We finish the proof of Lemma \ref{lemma_33}. \hfill $\square$
	\end{proof}
	
	\noindent Now we can calculate  $$\int_\Theta \pi_0(\theta) \exp \left (\theta^{\top} \sum_{i=1}^n {T}(\xi_i) - n A(\theta) \right )d\theta.$$  Note that without specific statement, all the following results hold almost surely ($\PP^\NN_{\theta^c}$). In fact, we have
	\begin{align*}
	&\int_\Theta \pi_0(\theta) \exp \left (\theta^\top \sum_{i=1}^n {T}(\xi_i) - n A(\theta) \right )d\theta\\
	&=\exp(\hat{\theta}_n^\top \sum_{i=1}^n {T}(\xi_i)-nA(\hat{\theta}_n))|(\nabla^2 A(\hat{\theta}_n))^{-\frac{1}{2}}|\int_{E_{\hat{\theta}_n}} \pi_0(\hat{\theta}_n+(\nabla^2 A(\hat{\theta}_n))^{-\frac{1}{2}}\theta)g^n(\theta,\frac{1}{n} \sum_{i=1}^n {T}(\xi_i))d\theta\\
	&=\exp(\hat{\theta}_n^\top \sum_{i=1}^n {T}(\xi_i)-nA(\hat{\theta}_n))|(\nabla^2 A(\hat{\theta}_n))^{-\frac{1}{2}}| (\frac{1}{\sqrt{n}})^s\pi_0(\boldsymbol \theta^c)\int_{\RR^s}\exp(-\frac{1}{2}\boldsymbol\eta^\top\boldsymbol\eta) d\boldsymbol\eta,
	\end{align*}
	where $|\nabla^2A(\hat{\theta}_n))^{-\frac{1}{2}}|$ is the determinant of $(\nabla^2A(\hat{\theta}_n))^{-\frac{1}{2}}.$
	To calculate 
	$$\int_\Theta \pi_0(\theta) \exp \left (\theta^{\top} \left(2\sum_{i=1}^{t_1} {T}(\xi_i)-\sum_{i=1}^{t_2} {T}(\xi_i)\right)-(2t_1-t_2) A(\theta) \right )d\theta$$
	for $t_1>t_2=t_1-k,$ where $k$ is a fixed positive constant, let's define the MLE  $\tilde{\theta}_{2t_1-t_2}$ satisfies  $$\nabla_\eta A(\tilde{\theta}_{2t_1-t_2})=\frac{1}{2t_1-t_2}\left(2\sum_{i=1}^{t_1} {T}(\xi_i)-\sum_{i=1}^{t_2} {T}(\xi_i)\right).$$  As $t_1\rightarrow \infty$, the right hand side converges almost surely to $\nabla_\eta A({\theta^c}),$ which implies $\tilde{\theta}_{2t_1-t_2}\rightarrow {\theta^c}$ almost surely. Then we have
	\begin{align*}
	& \int \pi_0(\theta) \exp \left (\theta^{\top} \left(2\sum_{i=1}^{t_1} {T}(\xi_i)-\sum_{i=1}^{t_2} {T}(\xi_i)\right)-(2t_1-t_2) A(\theta) \right )d\theta\\
	=&\exp(\tilde{\theta}_{2t_1-t_2}^\top \left(2\sum_{i=1}^{t_1} {T}(\xi_i)-\sum_{i=1}^{t_2} {T}(\xi_i)\right)-(2t_1-t_2)A(\tilde{\theta}_{2t_1-t_2}))|(\nabla^2 A(\tilde{\theta}_{2t_1-t_2}))^{-\frac{1}{2}}|\\
	&\qquad \cdot(\frac{1}{\sqrt{{2t_1-t_2}}})^s\pi_0(\boldsymbol \theta^c)\int_{\RR^s}\exp(-\frac{1}{2}\boldsymbol\eta^\top\boldsymbol\eta) d\boldsymbol\eta.
	\end{align*}
	Next, we can calculate the variance 
	\begin{align*}
	\EE\left(\frac{\pi_{t_1}(\theta)}{\pi_{t_2}(\theta) }\right)^2&=\frac{\left(\int_\Theta \exp \left (\theta^\top \boldsymbol\chi - t_2 A(\theta) \right )d\theta\right)\left( \int_\Theta  \exp \left (\theta^\top (2\boldsymbol\chi_{t_1}-\chi_{t_2}) - (2t_1-t_2)A(\theta) \right )d\theta\right)}{(\int_\Theta  \exp \left (\theta^\top \boldsymbol\chi' - t_1 A(\theta) \right )d\theta)^2}\\
	&=\frac{(\frac{1}{\sqrt{t_2}})^s(\frac{1}{\sqrt{{2t_1-t_2}}})^s}{(\frac{1}{\sqrt{t_1}})^{2s}}\frac{|\nabla^2 A(\hat{\theta}_{t_2}))^{-\frac{1}{2}}||\nabla^2 A(\tilde{\theta}_{2t_1-t_2}))^{-\frac{1}{2}}|}{|\nabla^2 A(\hat{\theta}_{t_1}))^{-\frac{1}{2}}|^2}\\
	&~~~~~~\times \frac{\exp(2\tilde{\theta}_{2t_1-t_2}^\top \sum_{i=1}^{t_1} {T}(\xi_i)-2{t_1}A(\tilde{\theta}_{2t_1-t_2}))}{\exp(2\hat{\theta}_{t_1}^\top \sum_{i=1}^{t_1} {T}(\xi_i)-2{t_1}A(\hat{\theta}_{t_1}))}\\
	&~~~~~~\times \frac{\exp(\hat{\theta}_{t_2}^\top \sum_{i=1}^{t_2} {T}(\xi_i)-{t_2}A(\hat{\theta}_{t_2}))}{\exp(\tilde{\theta}_{2t_1-t_2}^\top \sum_{i=1}^{t_2} {T}(\xi_i)-{t_2}A(\tilde{\theta}_{2t_1-t_2}))}.
	\end{align*}
	
	By the property of MLE, we know that
	$$\hat{\theta}_{t_2}^\top \sum_{i=1}^{t_2} {T}(\xi_i)-{t_2}A(\hat{\theta}_{t_2})\geq\tilde{\theta}_{2t_1-t_2}^\top \sum_{i=1}^{t_2} {T}(\xi_i)-{t_2}A(\tilde{\theta}_{2t_1-t_2}),$$
	$$\tilde{\theta}_{2t_1-t_2}^\top \left(\sum_{i=1}^{t_2} {T}(\xi_i)+2\sum_{i=t_2+1}^{t_2+k} {T}(\xi_i)\right)-(t_2+2k)A(\tilde{\theta}_{2t_1-t_2})\geq \hat{\theta}_{t_2}^\top \left(\sum_{i=1}^{t_2} {T}(\xi_i)+2\sum_{i=t_2+1}^{t_2+k} {T}(\xi_i)\right)-(t_2+2k)A(\hat{\theta}_{t_2}).$$
	The last two together implies 
	$$0\leq \hat{\theta}_{t_2}^\top \sum_{i=1}^{t_2} {T}(\xi_i)-{t_2}A(\hat{\theta}_{t_2})-\tilde{\theta}_{2t_1-t_2}^\top \sum_{i=1}^{t_2} {T}(\xi_i)-{t_2}A(\tilde{\theta}_{2t_1-t_2})\leq 2k(A(\hat{\theta}_{t_2})-A(\tilde{\theta}_{2t_1-t_2}))+2\sum_{i=t_2+1}^{t_2+k} {T}(\xi_i)(\tilde{\theta}_{2t_1-t_2}-\hat{\theta}_{t_2}).$$
	Since both MLE estimators converge to the true $\theta^c$ a.s., we have 
	$$2k(A(\hat{\theta}_{t_2})-A(\tilde{\theta}_{2t_1-t_2}))+2\sum_{i=t_2+1}^{t_2+k} {T}(\xi_i)(\tilde{\theta}_{2t_1-t_2}-\hat{\theta}_{t_2})\rightarrow 0 \text{ a.s.$(\PP^\NN_{\theta^c})$}.$$
	Thus 
	$$\frac{\exp(\hat{\theta}_{t_2}^\top \sum_{i=1}^{t_2} {T}(\xi_i)-{t_2}A(\hat{\theta}_{t_2}))}{\exp(\tilde{\theta}_{2t_1-t_2}^\top \sum_{i=1}^{t_2} {T}(\xi_i)-{t_2}A(\tilde{\theta}_{2t_1-t_2}))}\rightarrow 1 \text{ a.s.$(\PP^\NN_{\theta^c})$}.$$ 
	Similarly, we have
	$$\frac{\exp(2\tilde{\theta}_{2t_1-t_2}^\top \sum_{i=1}^{t_1} {T}(\xi_i)-2{t_1}A(\tilde{\theta}_{2t_1-t_2}))}{\exp(2\hat{\theta}_{t_1}^\top \sum_{i=1}^{t_1} {T}(\xi_i)-2{t_1}A(\hat{\theta}_{t_1}))}\rightarrow 1 \text{ a.s.$(\PP^\NN_{\theta^c})$},$$
	which further implies
	$$\EE\left(\frac{\pi_{t_1}(\theta)}{\pi_{t_2}(\theta) }\right)^2\rightarrow 1 \text{ a.s.$(\PP^\NN_{\theta^c})$},$$
	We finish the proof of Theorem \ref{thm:exp_out_bounded}. 
\end{proof}

\subsection{Proof of Corollary \ref{lem_bound} }
\begin{proof}
	Let $$\Omega'=\left\{\omega\in \Omega^\NN\Big|\EE\left(\frac{\pi_{t}(\theta)}{\pi_{\max\{t-k,0\}}(\theta) }\right)^2(\omega) \rightarrow 1, \text{~as~}t\rightarrow\infty\right\}.$$
	For any $\omega\in \Omega',$ since convergence implies boundedness, we know $\exists C(\omega)>0$, such that $$\EE\left(\frac{\pi_{t}(\theta)}{\pi_{ \max\{t-k,0\}}(\theta) }\right)^2<C,\forall 1\leq k\leq K, \forall t>0.$$ By Theorem \ref{thm:exp_out_bounded}, we have $$\PP^\NN_{\theta^c}\left(\exists C>0,\EE\left(\frac{\pi_{t}(\theta)}{\pi_{ \max\{t-k,0\}}(\theta) }\right)^2<C,\forall 1\leq k\leq K, \forall t>0.\right)\geq \PP^\NN_{\theta^c}(\Omega')=1.$$
	We finish the proof of Corollary \ref{lem_bound}.
\end{proof}
\subsection{Proof of Theorem \ref{thm:exp_in_bound}}
\begin{proof}
	We only need to show that $$\int \frac{p(x|\theta_1)^2}{p(x|\theta_2)}h(x)^2dx\leq C_1.$$
	Note that $\int\frac{p(x|\theta_1)^2}{p(x|\theta_2)}h(x)^2 dx=\int \kappa(x)\exp \left ( {(2\theta_1-\theta_2)}^\top{T}(x) -2A(\theta_1)+A(\theta_2)\right )h(x)^2dx.$ Since we have $2\theta_1-\theta_2\in\cN,$ we further have 
	\begin{align*}
	\int\frac{p(x|\theta_1)^2}{p(x|\theta_2)}h(x)^2 dx=\exp\left(-2A(\theta_1)+A(\theta_2)\right)\int  \kappa(x)h(x)^2\exp \left ((2\theta_1-\theta_2)^\top{T}(x)\right)dx.
	\end{align*}
	Note that $A(\theta)$ and $B_{h^2}(\theta)\coloneqq\log \int  \kappa(x)h(x)^2\exp \left (\theta^\top{T}(x)\right)dx$ are convex  and continuous in $\theta$. It is easy to see that $A(\theta)$ and $B_{h^2}(\theta)$ are bounded, which further implies that there exists $C_1>0$ such that 
	$$\int \frac{p(x|\theta_1)^2}{p(x|\theta_2)}h(x)^2dx\leq C_1.$$
	We finish the proof of Theorem \ref{thm:exp_in_bound}.
\end{proof}

\section{Proof of Theorem \ref{thm:consistency}}\label{proof:consistency}

\subsection{Proof of Lemma \ref{thm:in_consistency}}
\begin{proof} Recall that 
	$\hat{H}^{M,N}(\theta_\tau^s)	=\frac{1}{NM}\sum_{i=1}^M\sum_{j=1}^N\frac{p(\xi_0 ^{i,j}|\theta_\tau^s)}{p(\xi_0^{i,j}|\theta_\tau^i)}h(\xi_0^{i,j}).$
	For every $1\leq i\leq M,$ let $$\Omega_{i,\tau, s}=\left\{\omega\in\Omega: \lim_{N\rightarrow \infty}\frac{1}{N}\sum_{j=1}^N\frac{p(\xi_0 ^{i,j}(\omega)|\theta_\tau^s)}{p(\xi_0^{i,j}(\omega)|\theta_\tau^i)}h(\xi_0^{i,j}(\omega))={H}(\theta_\tau^s)\right\}. $$
	Then $$\left\{\forall \tau,s, \lim_{N\rightarrow \infty}\hat{H}^{M,N}(\theta_\tau^s)={H}(\theta_\tau^s)\right\}\supset\cup_{i,\tau,s}\Omega_{i,\tau, s}.$$ By the law of large number,
	$$\PP(\Omega_{i,\tau, s})=1.$$
	Then we have $$\PP\{\forall \tau,s>0, \lim_{N\rightarrow \infty}\hat{H}^{M,N}(\theta_\tau^s)={H}(\theta_\tau^s)\}=1.$$
	Let $\epsilon=\inf\left\{|H(\theta_{\tau_1}^{s_1})-H(\theta_{\tau_2}^{s_2})|\Big| t-k\leq \tau_1,\tau_2\leq t,~1\leq s_1,s_2\leq M\right\}.$
	Then for every $$\omega\in\left\{\forall \tau,s, \lim_{N\rightarrow \infty}\hat{H}^{M,N}(\theta_\tau^s)={H}(\theta_\tau^s)\right\},$$ there exists a $t_1>0$ such that when $N>t_1,$
	$$|\hat{H}^{M,N}(\theta_\tau^s)-{H}(\theta_\tau^s)|\leq \frac{\epsilon}{3},~\forall \tau,s>0.$$
	If $H(\theta_{\tau_1}^{s_1})<H(\theta_{\tau_2}^{s_2}),$ the performance estimator will keep this order, i.e.,
	$\hat{H}^{M,N}(\theta_{\tau_1}^{s_1})<\hat{H}^{M,N}(\theta_{\tau_2}^{s_2}).$ Let $H(\theta_t^{i_\alpha})=\hat{q}_{t}^{M,K}.$ Since the order is kept, we have   $\hat{H}^{M,N}(\theta_t^{i_\alpha})=\hat{q}_{t}^{M,N,K}.$ Then when $N>t_1,$
	$$|\hat{q}^{t,\alpha}_{N,M}-\hat{q}_{t}^{M,K}|=|\hat{H}^{M,N}(\theta_t^{i_\alpha})-H(\theta_t^{i_\alpha})|\leq \frac{\epsilon}{3},$$
	which implies $$\lim_{N\rightarrow\infty} \hat{q}_t^{M,N,K}=\hat{q}_{t}^{M,K}~~~~\text{ a.s.}$$
	We finish the proof of Lemma \ref{thm:in_consistency}. \hfill $\square$
\end{proof}

\subsection{Proof of Lemma \ref{thm:out_consistency}}
\begin{proof}
	For any $\delta>0,$ define a set $A_t^{M,K}(\delta)$ as follows.
	\begin{align*}
	A_t^{M,K}(\delta)&=\{\hat{q}_t^{M,K}\leq q_t-\delta\}\\&=\left\{\sum_{\tau=t-K+1}^t \sum_{j=1}^M \left(w_{t|\tau}^j \cI( H(\theta_\tau^i)\leq q_t-\delta)-G_t(q_t-\delta)\right)\geq KM\alpha-KMG_t(q_t-\delta)\right\}.
	\end{align*}
Let  $\sigma_{t|\tau}:=\sqrt{\Var{[w_{t|\tau}^j \cI( H(\theta_\tau^i)\geq q_t-\delta)]}}.$  Since $\EE[w_{t|\tau}^2]$ is uniformly bounded, there exists a constant $C>0$ such that   $\sigma_{t|\tau}\leq C.$ By the law of iterated logarithm,  for any $t-k\leq \tau\leq t,$
	$$\limsup_{M\rightarrow\infty}\frac{ \sum_{j=1}^M \left(w_{t|\tau}^j \cI( H(\theta_\tau^i)\geq q_t-\delta)-(1-G_t(q_t-\delta))\right)}{\sigma_{t|\tau}\sqrt{M\log\log M}}=1$$
	and
	$$\liminf_{M\rightarrow\infty}\frac{ \sum_{j=1}^M \left(w_{t|\tau}^j \cI( H(\theta_\tau^i)\geq q_t-\delta)-(1-G_t(q_t-\delta))\right)}{\sigma_{t|\tau}\sqrt{M\log\log M}}=-1.$$
	Define another set for $\eta>0$, \begin{align*}
	A_t^{M,K}(\delta,\eta)&=\left\{\sum_{\tau=t-K+1}^t \sum_{j=1}^M \left(w_{t|\tau}^j \cI( H(\theta_\tau^i)\leq q_t-\delta)-G_t(q_t-\delta)\right)\geq (1+\eta) KC\sqrt{M\log\log M}\right\}.
	\end{align*}
	Note that for large enough $M,$ we have $$KM\alpha-KMG_t(q_t-\delta)\geq (1+\eta) kC\sqrt{M\log\log M},$$ which implies $A_t^{M,K}(\delta)\subset A_t^{M,K}(\delta,\eta).$
	The law of iterated of logrithm implies, for any $\eta>0$
	$$\PP(A_t^{M,K}(\delta,\eta), \text{i.o.})=0\Rightarrow\PP(A_t^{M,K}(\delta), ~\text{i.o.})=0.$$
	Next, we show that $\hat{q}_t^{M,K}\geq q_t+\delta$ cannot happen infinitely often. Similarly,
	\begin{align*}
	B_t^{M,K}(\delta)&=\{\hat{q}_t^{M,K}\geq q_t+\delta\}\\
	&=\left\{\sum_{\tau=t-K+1}^t \sum_{j=1}^M \left(w_{t|\tau}^j \cI( H(\theta_\tau^i)\leq q_t+\delta)-G_t(q_t+\delta)\right)\leq KM\alpha-KMG_t(q_t+\delta)\right\}.
	\end{align*}
	Define another set for $\eta>0$, \begin{align*}
	B_t^{M,K}(\delta,\eta)&=\left\{\sum_{\tau=t-K+1}^t \sum_{j=1}^M \left(w_{t|\tau}^j \cI( H(\theta_\tau^i)\leq q_t+\delta)-G_t(q_t+\delta)\right)\leq-(1+\eta) kC\sqrt{M\log\log M}\right\}.
	\end{align*}
	Note that for large enough $M,$ we have $$KM\alpha-KMG_t(q_t+\delta)\leq -(1+\eta) kC\sqrt{M\log\log M},$$ which implies $B_t^{M,K}(\delta)\subset B_t^{M,K}(\delta,\eta).$
	The law of iterated of logrithm implies, for any $\eta>0,$
	$$\PP\left(B_t^{M,K}(\delta,\eta), \text{i.o.}\right)=0\Rightarrow\PP\left(B_t^{M,K}(\delta), ~\text{i.o.}\right)=0.$$
	Thus, we have for any $\delta>0$,
	$$\PP\left(A_t^{M,K}(\delta)\cup B_t^{M,K}(\delta), ~\text{i.o.}\right)=0,$$
	which implies $$\lim_{M\rightarrow\infty}\hat{q}_t^{M,K}= q_t,~~~~\text{ a.s.}$$
	We finish the proof of Lemma \ref{thm:out_consistency}.\hfill $\square$
\end{proof}

\section{Proof of Convergence Rate}\label{rate_proof}
\subsection{Proof of Theorem \ref{thm:out_rate}}\label{proof:out_rate}
\begin{proof}
The proof follows that of Theorem 1 in \cite{glynn1996importance}.	\begin{align*}
	&\PP(M^{\frac{1}{2}}(\hat{q}_t^{M,K}- q_t)\leq x)=\PP(\hat{q}_t^{M,K}\leq q_t+M^{-\frac{1}{2}}x)
	=\PP(\hat{G}_M(q_t+M^{-\frac{1}{2}}x)\geq \alpha)\\
	=&\PP\left(\frac{1}{KM}\sum_{\tau=t-K+1}^t \sum_{j=1}^M w_{t|\tau}^j \cI\left( H(\theta_\tau^i)\leq q_t+M^{-\frac{1}{2}}x\right)\geq \alpha\right)\\
	=&\PP\left(\frac{1}{M}\sum_{j=1}^M \frac{1}{k}\sum_{\tau=t-K+1}^t w_{t|\tau}^j \cI\left( H(\theta_\tau^i)\leq q_t+M^{-\frac{1}{2}}x\right)\geq \alpha\right)\\
	=&\PP\left(\frac{1}{M}\sum_{j=1}^M \frac{1}{k}\sum_{\tau=t-K+1}^t \left[w_{t|\tau}^j \cI\left( H(\theta_\tau^i)\leq q_t+M^{-\frac{1}{2}}x\right)-G_t(q_t+M^{-\frac{1}{2}}x)\right]\geq \alpha-G_t(q_t+M^{-\frac{1}{2}}x)\right).
	\end{align*}
	Let $s_t^{M,K}(x)=\sqrt{\frac{1}{K^2}\sum_{\tau=t-K+1}^t\EE\left\{  (w_{t|\tau}^j)^2 \cI\left( H(\theta_\tau^i)\leq q_t+M^{-\frac{1}{2}}x\right)-(G_t(q_t+M^{-\frac{1}{2}}x))^2\right\}}.$ By Berry-Esseen Theorem, we have 
	\begin{align}\label{eq16}
	&\left|\PP(M^{\frac{1}{2}}(\hat{q}_t^{M,K}- q_t)\leq x)-\PP\left(N(0,1)\geq M^{\frac{1}{2}}\frac{\alpha-G_t(q_t+M^{-\frac{1}{2}}x)}{s_t^{M,K}(x)}\right)\right|\nonumber\\
	&=\Bigg|\PP\left(N(0,1)\geq M^{\frac{1}{2}}\frac{\alpha-G_t(q_t+M^{-\frac{1}{2}}x)}{s_t^{M,K}(x)}\right)\nonumber\\&~~~~-\PP\left(M^{\frac{1}{2}}\frac{1}{M}\sum_{j=1}^M \frac{1}{K}\sum_{\tau=t-K+1}^t \left[w_{t|\tau}^j \cI\left( H(\theta_\tau^i)\leq q_t+M^{-\frac{1}{2}}x\right)-G_t(q_t+M^{-\frac{1}{2}}x)\right]\geq M^{\frac{1}{2}}(\alpha-G_t(q_t+M^{-\frac{1}{2}}x))\right)\Bigg|\nonumber\\
	&\rightarrow 0.
	\end{align}
	Moreover,  we have $$\lim_{M\rightarrow\infty}s_t^{M,K}(x)=\sqrt{\frac{1}{K^2}\sum_{\tau=t-K+1}^t\EE\left\{  (w_{t|\tau}^j)^2 \cI\left( H(\theta_\tau^i)\leq q_t\right)-\alpha^2\right\}},$$ and $$\lim_{M\rightarrow\infty}M^{\frac{1}{2}}(\alpha-G_t(q_t+M^{-\frac{1}{2}}x))=-xG'_t(q_t).$$
	If we further denote $$\tilde\sigma_t^K=\frac{\sqrt{\frac{1}{K^2}\sum_{\tau=t-K+1}^t\EE\left\{  (w_{t|\tau}^j)^2 \cI\left( H(\theta_\tau^i)\leq q_t\right)-\alpha^2\right\}}}{G'_t(q_t)},$$  then we have
	$$\lim_{M\rightarrow\infty}M^{\frac{1}{2}}\frac{\alpha-G_t(q_t+M^{-\frac{1}{2}}x)}{s_t^{M,K}(x)}= -\frac{1}{\tilde\sigma_t^K} x.$$ 
	Together with \eqref{eq16}, we can show that  $M^{\frac{1}{2}}(\hat{q}_t^{M,K}- q_t)$ is asymptotically normal.
	\begin{align*}
	&\left|\PP(M^{\frac{1}{2}}(\hat{q}_t^{M,K}- q_t)\leq x)-\PP\left(\tilde\sigma_t^K N(0,1)\leq x\right)\right|\\
	=&\left|\PP(M^{\frac{1}{2}}(\hat{q}_t^{M,K}- q_t)\leq x)-\PP\left(\tilde\sigma_t^K N(0,1)\geq -x\right)\right|\\
	=&\left|\PP(M^{\frac{1}{2}}(\hat{q}_t^{M,K}- q_t)\leq x)-\PP\left( N(0,1)\geq -\frac{1}{\tilde\sigma_t^K} x\right)\right|\rightarrow 0,
	\end{align*}
	which implies$$M^{\frac{1}{2}}(\hat{q}_t^{M,K}- q_t)\Rightarrow \tilde\sigma_t^KN(0,1),~~~\text{as $M\rightarrow\infty.$}$$
	This is equivalent to 
	$$(KM)^{\frac{1}{2}}(\hat{q}_t^{M,K}- q_t)\Rightarrow \sigma_t^KN(0,1),~~~\text{as $M\rightarrow\infty,$}$$
where $\sigma_t^K=\sqrt{\frac{1}{K}\sum_{\tau=t-K+1}^t\EE\left\{  (w_{t|\tau}^j)^2 \cI\left( H(\theta_\tau^i)\leq q_t\right)-\alpha^2\right\}}/G'_t(q_t).$
	
	Note that $\EE\left\{  (w_{t|\tau}^j)^2 \cI\left( H(\theta_\tau^i)\leq q_t\right)\right\}\leq\EE\left\{  (w_{t|\tau}^j)^2\right\},$ which is almost surely bounded according to Corollary \ref{lem_bound}. This further implies that $\sigma_t^{M,K}$ is almost surely bounded for all $t>0.$ 
	
	We finish the proof of Theorem \ref{thm:out_rate}.\hfill $\square$
\end{proof}
\subsection{Proof of Theorem \ref{thm:in_rate}}
\begin{proof}
	From the proof of  Theorem \ref{thm:consistency}, we know that for large enough $N,$ given $\theta_s^i,~ i=1,..,M, s=t-k,..,t,$ simulation uncertainty does not change the order statistics. Thus, suppose  $H(\theta_{\tau}^{i(\tau)})=\hat{q}_{t}^{M,K},$ then we have $$\lim_{N\rightarrow\infty} \sqrt{N}\left(\hat{q}_{t}^{M,N,K}-\hat{q}_{t}^{M,K}\right)=\lim_{N\rightarrow\infty} \sqrt{N}\left(\hat{H}^{M,N}\left(\theta_{\tau}^{i(\tau)}\right)-H\left(\theta_{\tau}^{i(\tau)}\right)\right).$$
	Recall that
	$\hat{H}^{M,N}(\theta_{\tau}^{i(\tau)})	=\frac{1}{NM}\sum_{i=1}^M\sum_{j=1}^N\frac{p(\xi_{\tau}^{i,j}|\theta_{\tau}^{i(\tau)})}{p(\xi_{\tau}^{i,j}|\theta_{\tau}^i)}h(\xi_{\tau}^{i,j}),$ where $\xi_{\tau}^{i,j}\sim p(\cdot|\theta_{\tau}^i), i=1,...,M, j=1,...,N$.
	By the central limit theorem, we have 
	\begin{align}\label{eq_in}
	\lim_{N\rightarrow\infty} \sqrt{N}\left(\hat{H}^{M,N}\left(\theta_{\tau}^{i(\tau)}\right)-H\left(\theta_{\tau}^{i(\tau)}\right)\right)=\sigma^M_{\tau}N(0,1) \text{ in distribution,}
	\end{align}
	where
	$$(\sigma^M_{\tau})^2=\Var\left\{\frac{1}{{M}}\sum_{i=1}^M\frac{p(\xi_{\tau} ^{i,j}|\theta_{\tau}^{i(\tau)})}{p(\xi_{\tau}^{i,j}|\theta_{\tau}^i)}h(\xi_{\tau}^{i,j})\Bigg| \theta_{\tau}^i,i=1,...,M\right\}=\frac{1}{M^2}\sum_{i=1}^M\Var\left\{\frac{p(\xi_{\tau} ^{i,j}|\theta_{\tau}^{i(\tau)})}{p(\xi_{\tau}^{i,j}|\theta_{\tau}^i)}h(\xi_{\tau}^{i,j})\Bigg| \theta_{\tau}^i,\theta_{\tau}^{i(\tau)}\right\}\leq \frac{C}{M},$$
	where $C$ is the constant defined in Theorem \ref{thm:exp_in_bound}.
	Thus, for $\beta\in[0,1)$,
	$$\lim_{M\rightarrow\infty} \frac{M^\beta}{M^2}\sum_{i=1}^M\Var\left\{\frac{p(\xi_{\tau} ^{i,j}|\theta_{\tau}^{i(\tau)})}{p(\xi_{\tau}^{i,j}|\theta_{\tau}^i)}h(\xi_{\tau}^{i,j})\Bigg| \theta_{\tau}^i,\theta_{\tau}^{i(\tau)}\right\}=0.$$ 
	This implies $$\lim_{M\rightarrow\infty}\sqrt{M^\beta}\sigma^M_{\tau}N(0,1)=0 \text{ in probability.}$$
	Together with \eqref{eq_in}, we prove Theorem \ref{thm:in_rate}. \hfill $\square$
\end{proof}
\subsection{Proof of Theorem \ref{thm:in_rate2}}
\begin{proof}
	Same as the case where CIS is applied, when $N$ is large enough, given $\theta_s^i,~ i=,..,M, s=t-k,..,t,$ simulation uncertainty does not change the order statistics. Thus,  we have $$\lim_{N\rightarrow\infty} \sqrt{N}\left(\breve{q}_{t}^{M,N,K}-\hat{q}_{t}^{M,K}\right)=\lim_{N\rightarrow\infty} \sqrt{N}\left(\breve{H}^{N}\left(\theta_{\tau}^{i(\tau)}\right)-H\left(\theta_{\tau}^{i(\tau)}\right)\right).$$
	By the central limit theorem, we have 
	\begin{align}\label{eq_in2}
	\lim_{N\rightarrow\infty} \sqrt{N}\left(\breve{H}^{N}\left(\theta_{\tau}^{i(\tau)}\right)-H\left(\theta_{\tau}^{i(\tau)}\right)\right)=\breve\sigma^M_{\tau}N(0,1), \text{ in distribution,}
	\end{align}
	where
	$$\left(\breve\sigma^M_{\tau}\right)^2=\Var_{\xi\sim p(\cdot|\theta_\tau^i(\tau))}\{h(\xi)\}.$$
	When $M\rightarrow\infty,$ we have $\lim_{M\rightarrow\infty}\Var_{\xi\sim p(\cdot|\theta_\tau^i(\tau))}\{h(\xi)\}=\Var_{\xi\sim p(\cdot|\theta_t)}\{h(\xi)\},$
	where $\theta_t$ is the $\theta-$parameter corresponding to the quantile $q_t,$ i.e., $H(\theta_t)=q_t.$ Then together with \eqref{eq_in2}, we have 
	$$\lim_{M\rightarrow\infty}\lim_{N\rightarrow\infty} \sqrt{N}\left(\breve{q}_{t}^{M,N,K}-\hat{q}_{t}^{M,K}\right)=\sigma_tN(0,1) \text{ in distribution, }$$
	where $\sigma_t^2=\Var_{\xi\sim p(\cdot|\theta_t)}\{h(\xi)\}.$
	We finish the proof of Theorem \ref{thm:in_rate2}. \hfill $\square$
\end{proof}

\subsection{Lower Boundedness of density quantile function $g_t(q_t)$}\label{bound_quantile_density}
Here, we consider the case that $\Theta\subset R$ is a compact set. Suppose $H(\theta)$ is strictly monotone and differentiable in $\Theta$. Note that $g_t$ is exactly the probability density function of $H(\theta_t),$ where $\theta_t\sim\pi_t.$  By the method of transformations,  the density quantile function can be calculated as follows:
$$g_t(q_t)=\frac{\pi_t(\theta_t)}{|H'(\theta_t)|},$$ where $\theta_t=H^{-1}(q_t)$ and $H'(\theta_t)$ is the derivative of $H$ at $\theta_t$. By strict monotonicity and smoothness of $H(\theta)$ and compactness of $\Theta,$ , we know that $|H'(\theta_t)|$ is lower bounded away from zero for all $t>0$.  Thus, we only need to show that $\pi_t(\theta_t)$ is lower bounded away from zero for all $t.$ Here we consider one special case that $\pi_t$ is the posterior distribution of  normal distribution with known variance $\sigma^2$, i.e.,
\begin{align}\label{gauss}
\pi_t\sim N\left(\frac{1}{\frac{1}{\sigma_0^2} + \frac{n}{\sigma^2}}\left(\frac{\mu_0}{\sigma_0^2} + \frac{\sum_{i=1}^n x_i}{\sigma^2}\right),\left(\frac{1}{\sigma_0^2} + \frac{n}{\sigma^2}\right)^{-1}\right),
\end{align}where $\mu_0,\sigma_0$ are respectively the expectation and variance of the prior normal distribution. Note that since $H$ is strictly monotone,  $H^{-1}(q_t)$ is either $\alpha$ or $1-\alpha$ quantile of $\pi_t.$ Without loss of generality, we assume $H^{-1}(q_t)=\theta_t^\alpha.$ Note that the $\alpha$ quantile $\theta^\alpha$ of a normal distribution $N(\mu,\sigma^2)$ (with p.d.f $\phi$ ) has an explicit form:
$$\theta^\alpha=\mu+\sqrt{2}\sigma \text{erf}^{-1}(2\alpha-1),$$
where $\text{erf}(x)=\frac{2}{\sqrt{\pi}}\int_0^xe^{-t^2}dt$ is the error function. Then  the density function of the quantile has the following form:
$$\phi(\theta^\alpha)=\frac{1}{\sqrt{2\pi}\sigma}\exp\left( \left(\text{erf}^{-1}\left(2\alpha-1\right)\right)^2\right).$$
Together with \eqref{gauss},
we know that $$\pi_t(\theta_t)=\frac{1}{\sqrt{2\pi\left(\frac{1}{\sigma_0^2} + \frac{n}{\sigma^2}\right)^{-1}}}\exp\left( \left(\text{erf}^{-1}\left(2\alpha-1\right)\right)^2\right),$$ which is strictly increasing to infinity and lower bounded by
$$\pi_0(\theta_0)=\frac{1}{\sqrt{2\pi\sigma_0^2}}\exp\left( \left(\text{erf}^{-1}\left(2\alpha-1\right)\right)^2\right)>0.$$
Thus $g_t(q_t)$ is lower bounded away from zero.

\section{Analysis for General Distributions}\label{discuss_general}
Though the previous analysis on EFDs is sufficient to justify the applicability of our algorithm, studying the condition required for general distributions still has its own merit. In this section, we provide one set of sufficient conditions such that Assumption \ref{assumption_Exp} holds.
The first assumption is on the input parameter.
\begin{assumption}\label{ass_3}
	~
	\begin{itemize}
		\item[1.] The input parameter space $\Theta$ is compact.
		\item[2.] For any neighborhood $V \in \mathcal{B}_{\theta}$ of $\theta^c$, there exists a sequence of uniformly consistent tests of the hypothesis $\tilde{\theta} = \theta^c$ against the alternative $\tilde{\theta} \in \Theta \setminus V$  .
		\item[3.] For any $\epsilon>0$ and any neighborhood $V \in \mathcal{B}_{\theta}$ of $\theta^c$, $V$ contains a subset $W$ such that $\pi(W)>0$ and $D_{KL}\left\{p(\cdot|{\theta^c}) \| p(\cdot|{\theta})\right\}<\epsilon$ for all $\theta \in W$.\end{itemize}
\end{assumption}
{The second item actually implies separability of $\theta^c$ from $\Theta \setminus V$. For more details on uniformly consistent tests, we refer the reader to \cite{schwartz1965bayes}.} We remark that Assumption \ref{ass_3} is also used in \cite{wu2018bayesian} (Assumption 3.1) to establish the strong consistency of posterior distributions, i.e., for any neighborhood  $V\in \cB_\theta$ of $\theta^c,$ $\int_V\pi_t(\theta)d\theta\rightarrow 1$ as $t\rightarrow \infty$ almost surely ($\PP_{\theta^c}^\NN$). Next assumption is on $\Omega$ and the likelihood function $p(\xi|\theta)$.
\begin{assumption}\label{ass_2}
	~\begin{itemize}
		\item[1.] $\Omega$ is compact.
		\item[2.] $p(\xi|\theta)$ is a continuous function in both $\xi$ and $\theta$ in $(\Omega,\Theta).$
	\end{itemize}
\end{assumption}
Assumption  \ref{ass_2} ensures $p(\xi|\theta)$ is uniformly continuous in $\xi$ for all $\theta\in \Theta.$ 	When $\Theta$  is finite, Assumption  \ref{ass_2} holds for every continuous likelihood function. Otherwise, we need further verify this assumption.  Under Assumptions \ref{ass_3} and \ref{ass_2}, we have the following theorem  for general distributions.	\begin{theorem}\label{thm:general_bound}
	Under Assumptions \ref{ass_3} and \ref{ass_2}, for a given $K>0,$  there exist constants $C_{3}>0$ and $C_4>0$ such that for any $1\leq k\leq K,$
	\begin{align*}
	&\EE_{\pi_{t-k}}\left[\left(\frac{\pi_t(\theta)}{\pi_{t-k}(\theta)}\right)^2\right]\leq C_{3},\\ &\EE_{\pi_{t-k}}\left[\left(\frac{\pi_t(\theta)}{\pi_{t-k}(\theta)}\right)^3\right]<\infty~, \\
	\sup_{\theta_1,\theta_2\in \Theta}&\Var\left\{\frac{p(\xi|\theta_1)}{p(\xi|\theta_2)}h(\xi)\Big| \theta_1,\theta_2\right\}\leq C_4, 
	\end{align*}
	where all the inequalities hold almost surely $(\PP^\NN_{\theta^c})$.
\end{theorem}
\begin{proof}
	Please refer to 	Appendix \ref{proof:general_bound}.
\end{proof}

Theorem \ref{thm:general_bound} shows that under Assumptions \ref{TLIS-1} and \ref{ass_2}, a general distribution enjoys the same good properties as EFDs shown in Theorem \ref{thm:exp_out_bounded}, Corollary \ref{lem_bound}, and Theorem \ref{thm:exp_in_bound}. Thus, following the same proofs, we can show that the conclusions in Theorems \ref{thm:consistency}, \ref{thm:out_rate}, and \ref{thm:in_rate} also hold for  general distributions.

\section{Proof of Theorem \ref{thm:general_bound}}\label{proof:general_bound}
\begin{proof}
	For notational simplicity, we  denote the likelihood function $\ell_t(\theta):=p(\xi_t|\theta).$  Note that according to \eqref{out_weight},  $\forall k\leq K,$
	$$\pi_{t}=\frac{ \Pi_{i=1}^{k}\ell_{t-k+i}(\theta) \pi_{t-k}(\theta)}{\int_\Theta\Pi_{i=1}^{k}\ell_{t-k+i}(\theta) \pi_{t-k}(\theta)d\theta},$$
	and
	$$\EE_{\pi_{t-k}}\left[\left(\frac{\pi_t(\theta)}{\pi_{t-k}(\theta)}\right)^2\right]=\int_\Theta \frac{\pi_t(\theta)^2}{\pi_{t-k}(\theta)}d\theta=\frac{\int_\Theta \Pi_{i=1}^{k}\ell_{t-k+i}^2 \pi_{t-k} d \theta}{\left(\int_\Theta\Pi_{i=1}^{k}\ell_{t-k+i}(\theta) \pi_{t-k}(\theta)d\theta\right)^2}.$$
	To show the second moment of the importance weight is bounded, we only need to upper bound $\int_\Theta \Pi_{i=1}^{k}\ell_{t-k+i}^2 \pi_{t-k} d \theta$ and lower bound $\int_\Theta\Pi_{i=1}^{k}\ell_{t-k+i}(\theta) \pi_{t-k}(\theta)d\theta.$  Under Assumption \ref{ass_3} and \ref{ass_2},  there exist constants $C_{min},~C_{max}>0,$ such that $C_{min}\leq p(\xi|\theta)\leq C_{max}.$  $p(\xi|\theta)$ is uniformly continuous, since a continuous function on a compact subset is uniformly continuous. 
	Then
	for any $\epsilon>0$, there exists $\delta>0$ such that when $\norm{\theta-\theta^c}_2\leq \delta,$ we have $|\ell_t(\theta)-\ell_t(\theta^c)|\leq \epsilon,~\forall t>0.$ 
	By the consistency of posterior distributions, we know $\pi_t(\theta)\rightarrow \delta(\theta^c)$. Moreover, for a fixed positive integer $n$, $\PP_{\theta\sim\pi_t}(\theta\in\Theta \backslash \Theta_n) \rightarrow 0,$ where $\Theta_n\subseteq \Theta$ is an open ball centered at $\theta^c$ with radius $1/n$.  Taking $n>1/\delta,$ then  there exists $\tau>\tau_1,$ such that when $t-K>\tau$,
 $\PP_{\theta\sim\pi_{t-k}}(\theta\in\Theta \backslash \Theta_n) \geq 1-\epsilon.$ We further have 
	\begin{align*}
	\int_\Theta \Pi_{i=1}^{k}\ell_{t-k+i}  \pi_{t-k}  d \theta&=\int_{\Theta\backslash\Theta_n} \Pi_{i=1}^{k}\ell_{t-+i} \pi_{t-k} d\theta+\int_{\Theta_n} \Pi_{i=1}^{k}\ell_{t-k+i} \pi_{t-k}  d\theta\\&\geq (1-\epsilon)\inf_{\Theta_n} \Pi_{i=1}^{k}\ell_{t-k+i}(\theta)\\&\geq(1-\epsilon) \Pi_{i=1}^{k}(\ell_{t-k+i}(\theta^c)-\epsilon).
	\end{align*}
	Following the similar argument, we have 
	\begin{align*}
	\int_\Theta \Pi_{i=1}^{k}\ell_{t-k+i}^2  \pi_{t-k}  d \theta&=\int_{\Theta\backslash\Theta_n} \Pi_{i=1}^{k}\ell_{t-k+i} ^2 \pi_{t-k}  d\theta+\int_{\Theta_n} \Pi_{i=1}^{k}\ell_{t-k+i}^2 \pi_{t-k}  d\theta\\
	&\leq (1-\epsilon)\sup_{\Theta_n} \Pi_{i=1}^{k}\ell_{t-k+i}^2(\theta)+C_{max}^{2k}\epsilon\\
	&\leq(1-\epsilon) \Pi_{i=1}^{k}(\ell_{t-k+i}(\theta^c)+\epsilon)^2+C_{max}^{2k}\epsilon.
	\end{align*}
	Thus, we have
	\begin{align*}
	\EE_{\pi_{t-k}}\left[\left(\frac{\pi_t(\theta)}{\pi_{t-k}(\theta)}\right)^2\right]&=\frac{\int_\Theta \Pi_{i=1}^{k}\ell_{t-k+i}^2  \pi_{t-k}  d \theta}{\left(\int_\Theta \Pi_{i=1}^{k}\ell_{t-k+i}  \pi_{t-k} d \theta\right)^2}\\
	&\leq\frac{(1-\epsilon) \Pi_{i=1}^{k}(\ell_{t-k+i}(\theta^c)+\epsilon)^2+C_{max}^{2k}\epsilon}{(1-\epsilon)^2 \Pi_{i=1}^{k}(\ell _{t-k+i}(\theta^c)-\epsilon)^2}.
	\end{align*}
	Note that  $\ell_t(\theta^c)$ is lower bounded away from 0. When $\epsilon\rightarrow 0$ (i.e., $t\rightarrow \infty$), the upper bound  goes to $1.$ Thus, there must exist a constant $C_3,$ such that for all $t>0,$ $$\EE_{\pi_{t-k}}\left[\left(\frac{\pi_t(\theta)}{\pi_{t-k}(\theta)}\right)^2\right]\leq C_{3}.$$
	Following the similar argument, we have
	$$\EE_{\pi_{t-k}}\left[\left(\frac{\pi_t(\theta)}{\pi_{t-k}(\theta)}\right)^3\right]<\infty.$$
	
	Finally, since $h$ is continuous and bounded,  $\frac{p(\xi|\theta_1)}{p(\xi|\theta_2)}h(\xi)$ is bounded, which implies there exists a constant $C_4>0,$ such that $$\sup_{\theta_1,\theta_2\in \Theta}\Var\left\{\frac{p(\xi|\theta_1)}{p(\xi|\theta_2)}h(\xi)\Big| \theta_1,\theta_2\right\}\leq C_4.$$
	We finish the proof of Theorem \ref{thm:general_bound}. \hfill $\square$

\end{proof}



\end{document}